\documentclass[final,5p,times,twocolumn]{elsarticle}

\usepackage{amssymb}
\usepackage{amsmath}
\usepackage{amsthm}
\usepackage{hyperref}

\newtheorem{theorem}{Theorem}[section]
\newtheorem{proposition}[theorem]{Proposition}

\newtheorem{corollary}[theorem]{Corollary}
\theoremstyle{definition}
\newtheorem{definition}[theorem]{Definition}

\theoremstyle{remark}
\newtheorem{remark}[theorem]{Remark}

\newcommand{\cH}{\mathcal H}
\newcommand{\cB}{\mathcal B}
\newcommand{\Tr}{\operatorname{Tr}}

\newcommand{\ee}{\mathrm e}
\newcommand{\ii}{\mathrm i}
\newcommand{\Y}[2]{Y_{#1,#2}}
\newcommand{\eps}{\varepsilon}
\makeatletter
\newcommand{\authororcid}[1]{%
  \ifx\corref\@gobble\else
    \texorpdfstring{\,\href{https://orcid.org/#1}{\raisebox{-0.15ex}{\includegraphics[height=1.7ex]{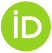}}}}{}%
  \fi}
\makeatother

\journal{Nuclear Physics B}

\begin{document}

\begin{frontmatter}



\title{Movable seams in root-of-unity XXZ chains:\\ Relative flux classes and antiunitary spectral pairing}

\author[a,c]{Boliang Yu\authororcid{0009-0006-4664-5068}\texorpdfstring{\corref{cor1}}{}}
\ead{yuboliang@sjtu.edu.cn}
\author[a]{Ruixin Zhou\authororcid{0009-0004-9303-6674}}

\author[b,c]{Meisen Gao\authororcid{0000-0002-4028-9895}}

\cortext[cor1]{Corresponding author}

\affiliation[a]{organization={School of Physics and Astronomy, Shanghai Jiao Tong University},
            city={Shanghai},
            postcode={200240}, 
            country={China}}
\affiliation[b]{organization={School of Physics, East China University of Science and Technology},
            city={Shanghai},
            postcode={200237}, 
            country={China}}
\affiliation[c]{organization={Shanghai Key Laboratory of Particle Physics and Cosmology},
            city={Shanghai},
            postcode={200240}, 
            country={China}}
\begin{abstract}
Starting from the cyclic duality tensors introduced by Vernier, Miao,
and Yamazaki (VMY), whose endpoint-traced matrix-product operators obey
$\mathbb Z_N$ Tambara--Yamagami fusion and realize topological defect
lines of the compactified-boson conformal field theory, we retain the
virtual endpoint as an $N$-state dynamical degree of freedom and
construct an exactly movable seam in the spin-$\tfrac12$ XXZ chain at
$q=\mathrm e^{\mathrm i\pi M/N}$ with $\gcd(M,N)=1$. The local movement
identity holds for any unitary $q$-Weyl pair; in the finite cyclic
realization, all output phases are locally gauge equivalent and share
the same four seam eigenvalues. On a ring, the output gauge becomes a
directed twist on a single bond, while endpoint conjugacy reduces the
$2N$ labels to two relative $\mathbb Z_2$ flux classes. An explicit
antiunitary symmetry protects the class
$\eta\equiv M-1\pmod 2$: for even $N$ it pairs distinct charge sectors
isospectrally, whereas for odd $N$ it fixes one sector and squares to
$-I$ there. Hence every many-body energy eigenspace in the protected
class has even multiplicity. The source value
$\eta_{\rm VMY}=1-M$ satisfies this condition for every coprime root,
whereas the ungauged value $\eta=0$ does so only for odd $M$. Exact
 finite-size controls show that the opposite class can contain simple
 levels, so a global flux invisible to local gauge equivalence distinguishes
 the two classes spectrally.
\end{abstract}


\end{frontmatter}

%\tableofcontents


\section{Introduction}

A phase attached to a local tensor is removable on an open segment.  On a
closed ring, removing it from the tensor places a twist on a bond, and the
remaining datum is a holonomy.  For the movable XXZ seam studied here this
elementary gauge fact is the beginning, rather than the end, of the problem:
the nontrivial question is which relative holonomy supports an explicit
projective spectral symmetry.

Matrix-product operators (MPOs) provide a microscopic language for dualities
and generalized symmetries in quantum chains
\citep{Lootens2023Dualities,Lootens2024Sectors,Miao2026MPO}.  When a local MPO
tensor reshapes into a unitary, a sequential circuit can move a localized
Hamiltonian modification; this dangling-virtual-space mechanism and related
defect-Hilbert-space constructions were developed in several recent works
\citep{Ueda2026Perfect,Vanhove2025Duality,Tantivasadakarn2025Sequential,Inamura2026Remarks,Wen2026NonInvertible,Sato2026Invariants}.
Vernier, Miao, and Yamazaki (VMY) constructed endpoint-traced root-of-unity
XXZ duality MPOs $D_\pm$ from a cyclic Lax tensor and noncommuting transfer
matrices \citep{Vernier2026Lattice}.  They showed that these operators obey
$\mathbb Z_N$ Tambara--Yamagami-type fusion rules and furnish a lattice
realization of topological defect lines in the compactified-boson conformal
field theory.  Their local cyclic blocks are the source of the construction
below, but our endpoint is retained rather than traced.  Consequently the
Hamiltonians studied here act on a different Hilbert space, and we do not
transfer the source MPO's fusion or continuum interpretation to them.

Here a \emph{seam} is a local replacement of one ordinary nearest-neighbor
XXZ bond by an operator on the two adjacent spins and an endpoint Hilbert
space; a local physical--endpoint unitary translates this replacement by one
site.  The endpoint is \emph{retained}: the MPO virtual space is kept as a
dynamical tensor factor instead of being contracted or traced.  By
\emph{output normalization} we mean a diagonal phase multiplying the
physical output leg of the fixed local reshape, not a normalization of states
or energies and not a virtual-bond similarity.  Its \emph{relative ring
holonomy} is the residual twist class left after comparison with a reference
output gauge and transport to an ordinary reference bond.

We first construct, for any unitary $q$-Weyl pair, two ordinary two-sided
gates and explicit Hermitian seams.  Their movement identity and quartic
polynomial are identities of the unrestricted Laurent--Weyl algebra.  At a
finite root, the generic-$M$ VMY blocks differ from the Weyl blocks only by a
branch-dependent output gauge.  We then determine exactly what that gauge
means on a ring: it can be transmuted into a single-bond twist, while endpoint
conjugation identifies $\eta$ with $\eta+2$.  The quotient is a relative
holonomy $\nu_\eta\in\{+1,-1\}$, not a collection of $2N$ independent global
objects.

The projective antiunitary independently selects the parity class
$\eta\equiv M-1\pmod2$.  The source value inherited from the VMY cyclic
blocks is $\eta_{\rm VMY}=1-M$ and therefore belongs to this class for every
coprime root, whereas the ungauged value $\eta=0$ does so only for odd $M$.
For even $N$ the antiunitary pairs distinct charge sectors isospectrally;
for odd $N$ it also fixes one sector and squares to $-I$ there.  At
$(N,M,L)=(3,2,3)$ the opposite orbit has eight simple levels, and the
sector-resolved finite-size scan supports the corresponding numerical
converse.

VMY established the Tambara--Yamagami fusion relation for their
endpoint-traced MPO at the source normalization.  Whether that source-side
fusion requirement is generically tied to the protected holonomy class in
movable-seam constructions remains an open question.

Section~\ref{sec:conventions} defines the finite and unrestricted objects;
Sec.~\ref{sec:unitarity} gives the local theorem;
Sec.~\ref{sec:flux} derives flux transmutation and the two holonomy orbits; and
Sec.~\ref{sec:globalpairing} derives the two sector-pairing mechanisms.
Exact and finite-size spectral controls follow in
Sec.~\ref{sec:controls}, with scope and physical interpretation collected in
Sec.~\ref{sec:limitations}.
\ref{app:unitarity} proves the Weyl-algebra and finite-root statements;
\ref{app:vmyprovenance} records the source map and flux transport;
\ref{app:ansatz} constructs the seam and derives its local spectrum;
\ref{app:uniqueness} proves support-qualified uniqueness;
\ref{app:global} gives the projective antiunitary; and
\ref{app:controls} collects exact controls, numerical diagnostics, and
reproduction conventions.

Figure~\ref{fig:result-overview} summarizes the physical chain from an output
gauge to a movable ring twist, the two relative-holonomy orbits, and the two
protected charge-pairing mechanisms.

Full-spectrum doubling, antiunitary pairing, and closed movement cycles have
precedents in other models or operator domains
\citep{Ueda2026Perfect,Liang2026Physical,Aasen2020Dualities,Sato2026Invariants};
root-of-unity XXZ multiplets and hidden sectors provide further, differently
scoped neighbors
\citep{DeguchiFabriciusMcCoy2001Loop,Deguchi2004Twisted,Korff2004Twisted,Miao2021QOperator,Hu2026Hidden}.
\ref{app:prior-boundaries} compares these neighboring mechanisms at
the level of their operator domains.

\section{Models and objects}
\label{sec:conventions}

\subsection{Notation and operator domains}
\label{sec:notation}

\begin{table*}[t]
\centering
\caption{Notation used throughout the article.}
\label{tab:notation}
\begin{tabular}{p{0.22\textwidth}p{0.72\textwidth}}
\hline
\textbf{Symbol} & \textbf{Meaning} \\
\hline
$\eps,q,\omega$ & $\eps=\ee^{\ii\pi/N}$, $q=\eps^M$, and $\omega=q^2$ at a finite root. \\
$X,Z$; $Y_{m,n}$ & Endpoint Weyl pair and Laurent monomial $Y_{m,n}=X^mZ^n$. \\
$\sigma$; $\chi_\sigma$ & Seam branch $\sigma\in\{+,-\}$ and branch factor $\chi_+=1$, $\chi_-=-1$. \\
$\cH_a$ & Retained $N$-state endpoint Hilbert space. \\
$f_\eta,H_\eta^\sigma$ & Physical-output phase $f_\eta=\eps^\eta$ and the corresponding retained-endpoint Hamiltonian. \\
$\eta_{\rm VMY}$ & Source output label $\eta_{\rm VMY}=1-M$. \\
$\nu_\eta$ & Relative $\mathbb Z_2$ ring-holonomy label. \\
$m_{\rm inv}$ & The residue $M^{-1}\bmod N$ in $\{0,\ldots,N-1\}$. \\
$\Gamma(t)$ & One-site diagonal gauge that transports a directed bond twist. \\
$Q$ & Combined physical--endpoint cyclic charge. \\
$\kappa_\eta,\tau_\eta$; $\delta$ & Compatible modular lifts and the parity bit $\delta=N\bmod2$. \\
$P_N,J,R,D_q(t)$ & Endpoint shear, endpoint inversion, spatial reflection, and diagonal physical gauge. \\
$E_{rc}$ & Matrix unit $|r\rangle\langle c|$ in the ordered two-spin basis. \\
$\mathcal C_\eta,\Theta_\eta,\mathcal A_\eta$ & Branch exchange, antiunitary branch map, and their same-branch composition. \\
$c_\eta^\sigma$ & Branch-dependent center of the sector involution $m\mapsto c_\eta^\sigma-m$. \\
$\mathcal B_{\eta,r}^\sigma$ & Same-branch projective antiunitary acting on $H_\eta^\sigma$. \\
\hline
\end{tabular}
\end{table*}

\begin{figure*}[t]
\centering
\includegraphics[width=\textwidth]{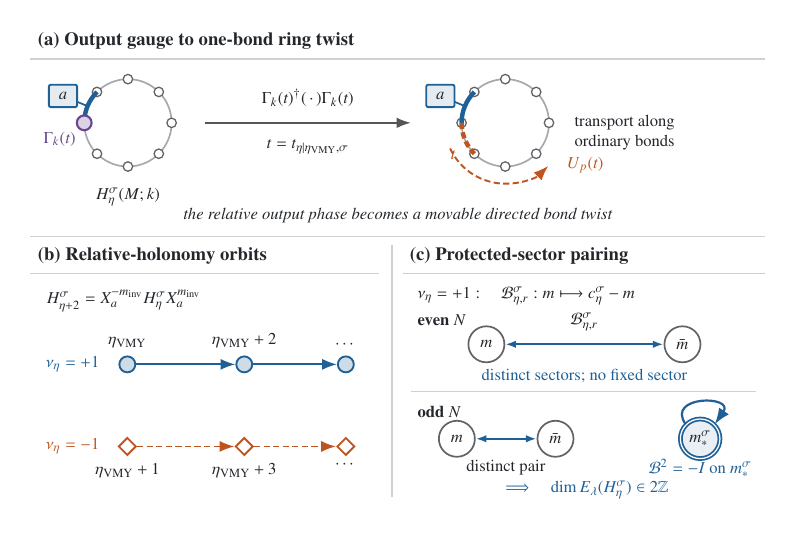}
\caption{From output gauge to protected spectral pairing.
(a) A one-site conjugation converts the relative output phase into a directed
bond twist, which moves along ordinary bonds.  (b) Endpoint conjugacy
$\eta\sim\eta+2$ leaves two relative-holonomy orbits.  (c) In the protected
orbit, even $N$ pairs distinct charge sectors; odd $N$ also gives a
fixed-sector Kramers pair.}
\label{fig:result-overview}
\end{figure*}

We reserve $s$ for the Lax parameter used in \citep{Vernier2026Lattice}
and write $\sigma\in\{+,-\}$ for the seam branch.  Let $q\in U(1)$,
$\omega=q^2$, and let $X,Z$ be bounded unitaries on the endpoint Hilbert
space $\cH_a$ satisfying
\begin{equation}
 ZX=\omega XZ.
 \label{eq:weyl}
\end{equation}
No finite-order condition is imposed here.  With
$Y_{m,n}=X^mZ^n$, our Laurent convention is
\begin{equation}
 Y_{m,n}Y_{r,t}=q^{2nr}Y_{m+r,n+t},\qquad
 Y_{m,n}^{\dagger}=q^{2mn}Y_{-m,-n}.
 \label{eq:weylproduct}
\end{equation}
The two physical--endpoint blocks and their ordinary-metric reshapes are
\begin{equation}
 \begin{aligned}
 L_+(q)&=\begin{pmatrix}I&\omega X^{-1}\\-qXZ&qZ\end{pmatrix},&
 L_-(q)&=\begin{pmatrix}qZ&-qX^{-1}Z\\X&\omega I\end{pmatrix},\\
 W_\sigma&=\frac{L_\sigma}{\sqrt2}.&&
 \end{aligned}
 \label{eq:Ltensors}
\end{equation}
On neighboring spins we use the Hermitian distributed-twist density
\begin{equation}
 h_{j,j+1}(q)=q\,\sigma_j^+\sigma_{j+1}^-
 +q^{-1}\sigma_j^-\sigma_{j+1}^+
 +\frac{q+q^{-1}}4\sigma_j^z\sigma_{j+1}^z.
 \label{eq:density}
\end{equation}

For the finite cyclic realization, set
\begin{equation}
 \begin{aligned}
 \eps&=\ee^{\ii\pi/N},\qquad q=\eps^M,\qquad N\ge2,\\
 &1\le M<N,\qquad \gcd(M,N)=1,
 \end{aligned}
 \label{eq:root}
\end{equation}
so $q^N=(-1)^M$ and $\omega$ has order $N$.  On
$\cH_a=\mathbb C^N$ we take
\begin{equation}
 Z|n\rangle=\omega^n|n\rangle,\qquad
 X|n\rangle=|n+1\!\!\pmod N\rangle,\qquad n=1,\ldots,N,
 \label{eq:weylrep}
\end{equation}
with residue zero represented by $|N\rangle$.  The periodic reference is
\begin{equation}
 H_{\rm unif}=\sum_{j=1}^{L}h_{j,j+1},\qquad L+1\equiv1.
 \label{eq:Hunif}
\end{equation}
It is gauge-equivalent to a boundary-twist presentation; all ring statements
below use Eq.~\eqref{eq:Hunif} and require $L\ge3$.

Taking the $s=0$ duality limit of the cyclic Lax tensor in
\citep{Vernier2026Lattice} gives, in our physical basis,
\begin{equation}
 \mathcal L^+_{\rm VMY}=\begin{pmatrix}I&\omega X^{-1}\\-\eps XZ&\eps Z\end{pmatrix},
 \qquad
 \mathcal L^-_{\rm VMY}=\begin{pmatrix}\eps Z&-\eps X^{-1}Z\\X&\omega I\end{pmatrix}.
 \label{eq:vmyBlocks}
\end{equation}
Here $q=\eps^M$ is the Weyl parameter of the generic cyclic relation,
whereas the displayed physical rows contain the principal phase $\eps$;
the two phases must therefore be kept distinct.
With
\begin{equation}
 f=\frac{\eps}{q},\qquad
 F_+=\operatorname{diag}(1,f),\qquad
 F_-=\operatorname{diag}(f,1),
 \label{eq:Fpm}
\end{equation}
the exact source map is
\begin{equation}
 \mathcal L_{\rm VMY}^{\sigma}=F_\sigma L_\sigma(q),\qquad
 W_\sigma^{\rm VMY}=F_\sigma W_\sigma.
 \label{eq:sourcegauge}
\end{equation}
Every diagonal one-site unitary $F$ also satisfies the bond-transport identity
\begin{equation}
 F_B^\dagger h_{BC}(q)F_B=F_Ch_{BC}(q)F_C^\dagger.
 \label{eq:bulktransport}
\end{equation}
Equations~\eqref{eq:sourcegauge} and \eqref{eq:bulktransport} are the entire
local bridge from the Weyl construction to the source-normalized one.  VMY's
transfer-matrix conservation is used only for their finite cyclic,
endpoint-traced representation; the unrestricted local Weyl theorem neither
requires nor constructs a traced transfer matrix.

Figure~\ref{fig:source-map} shows the source block, fixed reshape, and
retained-endpoint operator domains used in this construction.
\begin{figure*}[t]
\centering
\includegraphics[width=\textwidth]{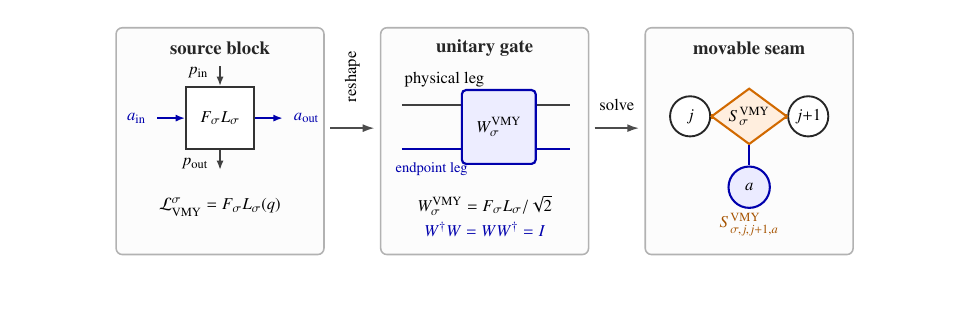}
\caption{From the source-normalized VMY block to the retained-endpoint seam.
The fixed physical--endpoint reshape gives a two-sided unitary gate, and the
local movement equation determines the two-spin seam.  Removing $F_\sigma$
gives the unrestricted Weyl construction.}
\label{fig:source-map}
\end{figure*}

Let $S_\sigma(q)$ denote the seam defined in
Eqs.~\eqref{eq:Splus}--\eqref{eq:Sminus}.  Retaining the endpoint $a$, define
\begin{align}
 H_{\rm Weyl}^{\sigma}(q;k)
 &=\sum_{j\neq k-1}h_{j,j+1}(q)+S_{\sigma,k-1,k,a}(q),
 \label{eq:HWeyl}\\
 S_{\sigma,ABa}^{\rm VMY}(M)
 &=F_{\sigma,B}S_{\sigma,ABa}(q)F_{\sigma,B}^\dagger,
 \label{eq:SVMYdef}\\
 H_{\rm VMY}^{\sigma}(M;k)
 &=\sum_{j\neq k-1}h_{j,j+1}+S_{\sigma,k-1,k,a}^{\rm VMY}.
 \label{eq:HVMY}
\end{align}
\begin{definition}[Output-normalization family]
\label{def:outputfamily}
Both retained-endpoint Hamiltonians sit in the family
\begin{align}
 f_\eta&=\eps^\eta,\nonumber\\[-0.4ex]
 F_+^{(\eta)}&=\operatorname{diag}(1,f_\eta),&
 F_-^{(\eta)}&=\operatorname{diag}(f_\eta,1),
 \label{eq:outputgaugefamily}\\
 W_\sigma^{(\eta)}&=F_\sigma^{(\eta)}W_\sigma,&
 S_\sigma^{(\eta)}&=F_{\sigma,B}^{(\eta)}S_\sigma
 F_{\sigma,B}^{(\eta)\dagger},\nonumber\\
 H_\eta^\sigma(M;k)&=\sum_{j\neq k-1}h_{j,j+1}
 +S_{\sigma,k-1,k,a}^{(\eta)},
 \label{eq:Heta}
\end{align}
where $\eta\in\mathbb Z_{2N}$.  Thus
$H_0^\sigma=H_{\rm Weyl}^\sigma$ and
$H_{1-M}^\sigma=H_{\rm VMY}^\sigma$.
\end{definition}

Bond transport makes all members of Definition~\ref{def:outputfamily}
locally gauge equivalent.  Section~\ref{sec:flux} identifies the global datum
that survives on a ring, after the local movement theorem has been stated.

Here ``source-normalized'' denotes the convention obtained from
Eq.~\eqref{eq:vmyBlocks} after fixing the clock--shift basis,
ordinary-metric reshape, distributed bulk, two-spin-plus-endpoint support,
and traceless affine representative.

\subsection{Objects and scope of claims}
\label{sec:objects-scope}

The article uses four operator scopes.  (i) The unrestricted local identities
hold for bounded unitary $q$-Weyl pairs.  (ii) Flux transmutation, endpoint
conjugacy, and spectral pairing concern the finite cyclic retained-endpoint
Hamiltonians on $(\mathbb C^2)^{\otimes L}\otimes\mathbb C_a^N$.  (iii) The
endpoint-traced VMY operators enter through the reproduced local blocks, while
their fusion and continuum interpretation belongs to the traced operator
domain.  (iv) Exact controls and the finite scan test the opposite orbit;
transfer and second-moment calculations supply additional diagnostics.
Endpoint-Weyl conjugation divides the labels into two parity orbits, implying
at most two classes under unrestricted global unitary equivalence.  The
analytic converse and a dressed-integrability construction remain open.

Table~\ref{tab:four-objects} compares the four operator domains used in
the article.
\begin{table*}[t]
\centering
\caption{Operator domains and global statements used in this work.}
\label{tab:four-objects}
\begin{tabular}{p{0.22\textwidth}p{0.30\textwidth}p{0.40\textwidth}}
\hline
\textbf{Object} & \textbf{Local structure} & \textbf{Global statement} \\
\hline
\textbf{Unrestricted $q$-Weyl seam}\newline
\emph{endpoint algebra}
& Two-sided unitary gate;\newline exact one-step movement;\newline quartic local identity.
& Algebraic scope only; no finite root or traced transfer matrix is assumed. \\
\hline
\textbf{Ungauged $H_{\mathrm{Weyl}}^{\sigma}$}\newline
\emph{endpoint retained}
& Finite clock--shift evaluation;\newline relative to VMY by a one-bond twist;\newline the same local seam spectrum.
& $\nu=(-1)^{M-1}$.\newline Odd $M$: protected; even $M$: opposite holonomy orbit. \\
\hline
\textbf{Source-normalized $H_{\mathrm{VMY}}^{\sigma}$}\newline
\emph{endpoint retained}
& VMY output gauge and fixed reshape;\newline the same movement law and seam spectrum.
& The source value satisfies $\eta\equiv M-1\pmod2$ for every coprime root; even/odd $N$ give the two pairing mechanisms. \\
\hline
\textbf{VMY temporal MPO}\newline
\emph{endpoint traced}
& Endpoint-traced source operator\newline on the physical chain.
& Transfer conservation belongs to the finite cyclic source representation; the present pairing theorem concerns the retained family. \\
\hline
\end{tabular}
\end{table*}

\section{Local movable-seam theorem}
\label{sec:unitarity}
\label{sec:construction}
\label{sec:movement}

The local statement separates the algebraic construction from the finite
root-of-unity consequences.  Let $A,B,C$ be successive physical sites.  The
seam is required to obey
\begin{equation}
 W_{\sigma,Ba}^{\dagger}
 \bigl(h_{BC}+S_{\sigma,ABa}\bigr)W_{\sigma,Ba}
 =h_{AB}+S_{\sigma,BCa}.
 \label{eq:movement}
\end{equation}
In the ordered two-spin basis
$|0\rangle=|00\rangle$, $|1\rangle=|01\rangle$,
$|2\rangle=|10\rangle$, $|3\rangle=|11\rangle$, write
$E_{rc}=|r\rangle\langle c|$.

\begin{theorem}[Local movable Weyl seam]
\label{thm:localmovement}
For every unitary representation of Eq.~\eqref{eq:weyl},
\begin{equation}
 W_\sigma^\dagger W_\sigma=W_\sigma W_\sigma^\dagger
 =I_{\mathbb C^2\otimes\cH_a},\qquad \sigma=\pm .
 \label{eq:unitarity}
\end{equation}
The following Hermitian Laurent operators solve Eq.~\eqref{eq:movement}:
\begin{align}
S_+(q)=\frac12\bigl[&
-E_{01}\Y{-1}{-1}+q^3E_{02}\Y{-1}{0}
-q^4E_{03}\Y{-2}{-1}\nonumber\\
&-q^2E_{10}\Y{1}{1}+q^2E_{12}\Y{0}{1}
-qE_{13}\Y{-1}{0}\nonumber\\
&+q^{-3}E_{20}\Y{1}{0}+q^{-2}E_{21}\Y{0}{-1}
+q^2E_{23}\Y{-1}{-1}\nonumber\\
&-E_{30}\Y{2}{1}-q^{-1}E_{31}\Y{1}{0}
+E_{32}\Y{1}{1}\bigr],
\label{eq:Splus}\\
S_-(q)=\frac12\bigl[&
E_{01}\Y{-1}{1}-qE_{02}\Y{-1}{0}
-E_{03}\Y{-2}{1}\nonumber\\
&+q^{-2}E_{10}\Y{1}{-1}+q^2E_{12}\Y{0}{-1}
+q^3E_{13}\Y{-1}{0}\nonumber\\
&-q^{-1}E_{20}\Y{1}{0}+q^{-2}E_{21}\Y{0}{1}
-q^{-2}E_{23}\Y{-1}{1}\nonumber\\
&-q^{-4}E_{30}\Y{2}{-1}+q^{-3}E_{31}\Y{1}{0}
-E_{32}\Y{1}{-1}\bigr].
\label{eq:Sminus}
\end{align}
They have vanishing physical diagonal and physical partial trace; every
finite-dimensional evaluation is traceless.
\end{theorem}

\begin{proof}[Proof sketch]
The four blocks of $L_\sigma^\dagger L_\sigma$ and
$L_\sigma L_\sigma^\dagger$ reduce to the Weyl relation and give $2I$.
For the seam, expand the complete fixed-$N$ operator space in
$E_{rc}\otimes X^mZ^n$.  A representation-independent particular solution is
found in the source-generated rectangle $-2\le m\le2$, $-1\le n\le1$.
The physical matrix-unit multiplication is
\[
 E_{sr}E_{j\ell}E_{\ell w}=\delta_{rj}E_{sw},
\]
so the coefficient map includes the factor $\delta_{rj}$ and vanishes when
$r\ne j$.  Exact Laurent normal ordering yields
Eqs.~\eqref{eq:Splus}--\eqref{eq:Sminus}; substitution cancels every
coefficient of the movement residual.  The full block multiplication,
coefficient map, and cancellation organization are given in
\ref{app:unitarity} and \ref{app:ansatz}.  The complete homogeneous
kernel is classified below, independently of the Laurent rectangle.
\end{proof}

\begin{proposition}[Finite-dimensional root constraint]
If invertible $d\times d$ matrices satisfy $ZX=q^2XZ$, then
\begin{equation}
 (q^2)^d=1.
 \label{eq:detconstraint}
\end{equation}
If $q^2$ has exact order $N$, every finite-dimensional irreducible unitary
representation has dimension $N$ and is, up to phases and unitary equivalence,
the clock--shift representation in Eq.~\eqref{eq:weylrep}.
\end{proposition}

\begin{proof}[Proof sketch]
The determinant of the Weyl relation gives Eq.~\eqref{eq:detconstraint}.
For the irreducible statement, the orbit of a $Z$ eigenvector under $X$ has
$N$ distinct eigenvalues and spans an invariant subspace.  Details are in
\ref{app:unitarity}.
\end{proof}

\begin{theorem}[Finite source-normalized movement]
\label{thm:genericMmovement}
For every coprime $(M,N)$ and either branch, the source-normalized seam is
Hermitian and traceless and obeys
\begin{equation}
 (W^{\rm VMY}_{\sigma,Ba})^\dagger
 \bigl[h_{BC}+S^{\rm VMY}_{\sigma,ABa}\bigr]
 W^{\rm VMY}_{\sigma,Ba}
 =h_{AB}+S^{\rm VMY}_{\sigma,BCa}.
 \label{eq:genericMmovement}
\end{equation}
On a distributed-twist ring with $L\ge3$, if $G_{\sigma,k}^{(\eta)}$ embeds
$W_\sigma^{(\eta)}$ at site $k$, then every output normalization satisfies
\begin{equation}
 (G_{\sigma,k}^{(\eta)})^\dagger
 H_\eta^\sigma(M;k)G_{\sigma,k}^{(\eta)}
 =H_\eta^\sigma(M;k+1).
 \label{eq:globalmovement}
\end{equation}
In particular,
\begin{align}
 G_{\sigma,k}^\dagger H_{\rm Weyl}^\sigma(q;k)G_{\sigma,k}
 &=H_{\rm Weyl}^\sigma(q;k+1),
 \label{eq:globalmovementWeyl}\\
 (G_{\sigma,k}^{\rm VMY})^\dagger H_{\rm VMY}^\sigma(M;k)
 G_{\sigma,k}^{\rm VMY}
 &=H_{\rm VMY}^\sigma(M;k+1).
 \label{eq:globalmovementVMY}
\end{align}
\end{theorem}

\begin{proof}[Proof sketch]
Insert $W_\sigma^{(\eta)}=F_\sigma^{(\eta)}W_\sigma$ and
$S_\sigma^{(\eta)}=F_{\sigma,B}^{(\eta)}S_\sigma
F_{\sigma,B}^{(\eta)\dagger}$ into Eq.~\eqref{eq:movement}; the diagonal
transport identity \eqref{eq:bulktransport} moves the output gauge to site
$C$.  All spectator bonds commute with the embedded gate.  This proves the
local and ring identities, including the three-distinct-site case $L=3$.
\end{proof}

Figure~\ref{fig:movement} depicts the one-step identity and the ordinary bond
restored behind the translated seam.
\begin{figure*}[t]
\centering
\includegraphics[width=\textwidth]{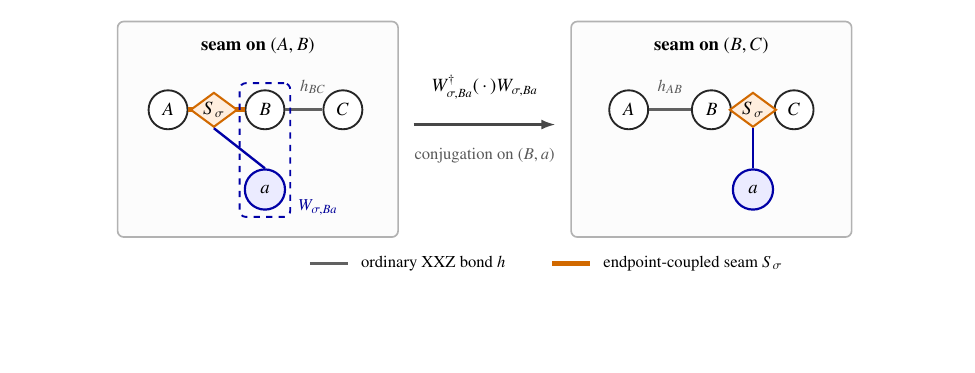}
\caption{One-step movement of the retained-endpoint seam.
The dashed box denotes only the support $B\otimes a$ of the gate
$W_{\sigma,Ba}$; it does not denote the full seam support or the identity on
site $C$.  Conjugation moves the seam from $ABa$ to $BCa$ while restoring
the ordinary bond on $AB$.}
\label{fig:movement}
\end{figure*}

\subsection{Cyclic charge and local spectrum}
\label{sec:symmetry}

Define the combined physical--endpoint charge
\begin{equation}
 g=\operatorname{diag}(1,\omega),\qquad
 Q=g^{\otimes L}\otimes Z_a^{-1}.
 \label{eq:basicQ}
\end{equation}
The selection rule of every seam monomial makes
\begin{equation}
 [H_\eta^\sigma(M;k),Q]=0
 \qquad(\eta\in\mathbb Z_{2N},\ \sigma=\pm).
 \label{eq:chargecomm}
\end{equation}
This is a single combined $\mathbb Z_N$ charge; it is not conservation of
physical $S^z$ alone.

\begin{theorem}[Local seam spectrum]
In the unrestricted Laurent--Weyl algebra,
\begin{equation}
 S_\sigma^4-\frac32S_\sigma^2+\frac{q+q^{-1}}2S_\sigma
 +\left[\frac1{16}-\frac1{16}(q+q^{-1})^2\right]I=0.
 \label{eq:seamquartic}
\end{equation}
In the standard $N$-dimensional realization, with
$\theta=\pi M/N$, all output-normalized seams have
\begin{equation}
 \operatorname{spec}S_\sigma^{(\eta)}
 =\left\{-\frac12\pm\cos\frac\theta2,\quad
          \frac12\pm\sin\frac\theta2\right\},
 \label{eq:seamspectrum}
\end{equation}
and each of the four distinct eigenvalues has multiplicity $N$.
\end{theorem}

\begin{proof}[Proof sketch]
Normal ordering verifies the quartic coefficient by coefficient.  In the
finite representation, the local restriction of $Q$ splits the $4N$-dimensional
space into $N$ four-dimensional blocks.  Diagonal gauges remove the charge
label, and the common block characteristic polynomial is
\[
 \lambda^4-\frac32\lambda^2+\cos\theta\,\lambda
 +\frac1{16}-\frac{\cos^2\theta}{4}.
\]
Its factorization gives Eq.~\eqref{eq:seamspectrum}.  Output gauges act by
unitary similarity.  The unrestricted Laurent verification and the finite
block matrices are given in \ref{app:localquartic} and
\ref{app:localblocks}, respectively.
\end{proof}

\subsection{Support-class uniqueness}
\label{sec:uniqueness}

The Laurent rectangle used above constructs a particular solution; uniqueness
is instead a statement in the complete local operator algebra.

\begin{theorem}[Support-class homogeneous kernel]
Let $K$ be one translationally identical bounded tensor on two neighboring
physical sites and the fixed endpoint representation.  If
\begin{equation}
 W_{\sigma,Ba}^\dagger K_{ABa}W_{\sigma,Ba}=K_{BCa},
 \label{eq:homogeneous}
\end{equation}
then
\begin{equation}
 K=I_A\otimes I_B\otimes r_a,\qquad r_a\in\{X,Z\}'.
 \label{eq:abstractkernel}
\end{equation}
Conversely every such $r_a$ is a homogeneous solution.  Hence the complete
affine family is
\begin{equation}
 S_{\sigma,\mathrm{general}}
 =S_\sigma+I_A\otimes I_B\otimes r_a,\qquad r_a\in\{X,Z\}'.
 \label{eq:abstractaffinefamily}
\end{equation}
For the irreducible clock--shift endpoint this reduces to
\begin{equation}
 \begin{aligned}
 S_{\sigma,\mathrm{general}}&=S_\sigma+cI_{4N},\\
 c&\in\mathbb R\quad\hbox{for Hermitian seams},\\
 c&=0\quad\hbox{after traceless normalization}.
 \end{aligned}
 \label{eq:affinefamily}
\end{equation}
The same statement holds for every output-normalized family by conjugation.
\end{theorem}

\begin{proof}[Proof sketch]
Expand $K$ in an orthonormal operator basis on site $A$.  Since
$W_{\sigma,Ba}$ does not act on $A$, every traceless $A$ coefficient must
vanish.  Translating the same tensor and using
\[
 [\mathcal B(Ba)\otimes I_C]\cap
 [I_B\otimes\mathcal B(Ca)]
 =I_B\otimes I_C\otimes\mathcal B(a)
\]
leaves an endpoint operator $r_a$.  Equation~\eqref{eq:homogeneous} then
requires $[r_a,X]=[r_a,Z]=0$; Schur's lemma makes it scalar in the irreducible
finite representation.  The full support-intersection proof is given in
\ref{app:uniqueness}.
\end{proof}

The theorem concerns the stated two-spin-plus-endpoint support, one translated
tensor, fixed representation and reshape.  Wider support, site dependence, or
a different endpoint representation defines a different classification
problem.

\section{Flux transmutation and output orbits}
\label{sec:flux}

Define the branch factors and diagonal output gauge by
\begin{equation}
 \chi_+=1,\qquad \chi_-=-1,\qquad
 \Gamma(t)=\operatorname{diag}(1,t),
 \label{eq:branchsigngauge}
\end{equation}
and, for a directed bond, set
\begin{equation}
 \begin{aligned}
 h_{j,j+1}^{[t]}&:=\Gamma_j(t)^\dagger h_{j,j+1}\Gamma_j(t)\\
 &=qt\,\sigma_j^+\sigma_{j+1}^-
  +q^{-1}t^{-1}\sigma_j^-\sigma_{j+1}^+\\
 &\quad+\frac{q+q^{-1}}4\sigma_j^z\sigma_{j+1}^z .
 \end{aligned}
 \label{eq:twistedbond}
\end{equation}

\begin{theorem}[Flux transmutation]
\label{thm:fluxreduction}
Let $L\ge3$ and choose two output labels
$\eta,\eta_0\in\mathbb Z_{2N}$.  Put
\begin{equation}
 r_{\eta\mid\eta_0}=\frac{f_\eta}{f_{\eta_0}},\qquad
 t_{\eta\mid\eta_0,\sigma}=r_{\eta\mid\eta_0}^{\chi_\sigma}.
 \label{eq:relativetwist}
\end{equation}
If the seam occupies $(k-1,k)$, then
\begin{equation}
 \Gamma_k(t)^\dagger H_\eta^\sigma(M;k)\Gamma_k(t)
 =H_{\eta_0}^\sigma(M;k)-h_{k,k+1}+h_{k,k+1}^{[t]},
 \qquad t=t_{\eta\mid\eta_0,\sigma}.
 \label{eq:fluxtransmutation}
\end{equation}
For $p=0,\ldots,L-2$, define
\begin{equation}
 U_p(t)=\prod_{\ell=0}^{p}\Gamma_{k+\ell}(t)^\dagger,
 \qquad b_p=k+p\pmod L .
 \label{eq:fluxstring}
\end{equation}
Then
\begin{equation}
 U_p(t)H_\eta^\sigma(M;k)U_p(t)^\dagger
 =H_{\eta_0}^\sigma(M;k)-h_{b_p,b_p+1}
  +h_{b_p,b_p+1}^{[t]}.
 \label{eq:fluxatbond}
\end{equation}
Thus the relative output gauge is exactly a single directed bond twist that
can be transported along the ordinary part of the ring.
\end{theorem}

\begin{proof}[Outline of the proof]
For the two branches,
$F_+^{(\eta)}=\Gamma(f_\eta)$ and
$F_-^{(\eta)}=f_\eta\Gamma(f_\eta^{-1})$; the scalar in the second identity
drops out of conjugation.  Removing the relative output factor changes only
the adjacent ordinary bond, and repeated use of
Eq.~\eqref{eq:bulktransport} transports it.  \ref{app:vmyprovenance}
gives the complete termwise proof.
\end{proof}

\begin{corollary}[Endpoint conjugacy and relative holonomy]
\label{cor:outputorbits}
Let $m_{\rm inv}\in\{0,\ldots,N-1\}$ be the unique residue satisfying
$Mm_{\rm inv}\equiv1\pmod N$.  For both branches and every seam position,
\begin{equation}
 H_{\eta+2}^{\sigma}(M;k)
 =X_a^{-m_{\rm inv}}H_\eta^\sigma(M;k)X_a^{m_{\rm inv}}.
 \label{eq:etaplustwo}
\end{equation}
Relative to $\eta_{\rm VMY}=1-M$, the endpoint-Weyl orbits are labelled by
\begin{equation}
 \nu_\eta:=\left(\frac{f_\eta}{f_{\rm VMY}}\right)^N
 =(-1)^{\eta-(1-M)}\in\{+1,-1\}.
 \label{eq:relativeflux}
\end{equation}
The oriented twist of Eq.~\eqref{eq:relativetwist} obeys
$t^N=\nu_\eta$ when $\eta_0=\eta_{\rm VMY}$.
\end{corollary}

\begin{proof}[Outline of the proof]
The seam selection rule ties the $Z$ exponent in $Y_{u,v}=X^uZ^v$ to the
output-spin matrix-unit difference.  Since
\[
 X_a^{-m_{\rm inv}}Y_{u,v}X_a^{m_{\rm inv}}
 =(q^2)^{m_{\rm inv}v}Y_{u,v}=\varepsilon^{2v}Y_{u,v},
\]
endpoint conjugation supplies exactly $f_\eta\mapsto f_{\eta+2}$.
\ref{app:vmyprovenance} gives the full calculation.
\end{proof}

\begin{remark}[Charge labels under endpoint conjugacy]
\label{rem:chargeshift}
The equivalence in Eq.~\eqref{eq:etaplustwo} is not sector-label preserving.
With $Q=g^{\otimes L}\otimes Z_a^{-1}$,
\begin{equation}
 X_a^{-m_{\rm inv}}QX_a^{m_{\rm inv}}
 =\omega^{-m_{\rm inv}}Q .
 \label{eq:chargeshift}
\end{equation}
Thus the forward identification from $H_\eta^\sigma$ to
$H_{\eta+2}^\sigma$ shifts $m$ to $m+m_{\rm inv}$ modulo $N$.  This agrees
with $\kappa_{\eta+2}=\kappa_\eta-2m_{\rm inv}$ and
$c_{\eta+2}^\sigma=c_\eta^\sigma+2m_{\rm inv}$ modulo $N$.
\end{remark}

\begin{remark}[Portability of flux transmutation]
The transmutation criterion is portable: whenever a movable
seam has a diagonal physical-output freedom whose relative factor can be
removed by one physical-site conjugation, and the ordinary bond obeys the
corresponding transport identity, the output freedom becomes a one-bond ring
twist.  The number and protection of the resulting classes still depend on
the endpoint representation and on additional symmetry identities.
\end{remark}

\section{Global movement and antiunitary pairing}
\label{sec:globalpairing}

Local gauge equivalence does not determine the global spectrum.  The movement
gates generate a closed-cycle symmetry, while a branch-exchange map composed
with complex conjugation produces an antiunitary whose projective square
depends on the output-normalization class.  Its complete dependency chain is
proved in \ref{app:global}.

\subsection{Holonomy and co-moving translation}

Label sites by $\mathbb Z_L$, place the seam on $(L-1,0)$, and let $T$
translate states one site to the right.  For the output-normalized gate define
\begin{equation}
 \mathcal U_{\eta,\sigma}(k)=
 G_{\sigma,k}^{(\eta)}G_{\sigma,k+1}^{(\eta)}\cdots
 G_{\sigma,k+L-1}^{(\eta)},\qquad
 \mathcal M_{\eta,\sigma}=G_{\sigma,0}^{(\eta)}T ,
 \label{eq:holonomy}
\end{equation}
with indices modulo $L$ and rightmost factors acting first.

\begin{proposition}[Movement holonomy]
For every coprime $(M,N)$, $L\ge3$, output exponent $\eta$, and branch
$\sigma$,
\begin{equation}
 \begin{aligned}
 \relax[\mathcal U_{\eta,\sigma}(k),H_\eta^\sigma(M;k)]
 &=[\mathcal U_{\eta,\sigma}(k),Q]=0,\\
 [\mathcal M_{\eta,\sigma},H_\eta^\sigma(M;0)]
 &=[\mathcal M_{\eta,\sigma},Q]=0,\\
 \mathcal M_{\eta,\sigma}^{L}&=\mathcal U_{\eta,\sigma}(0).
 \end{aligned}
 \label{eq:holonomycomm}
\end{equation}
\end{proposition}

\begin{proof}[Proof sketch]
Iterating Eq.~\eqref{eq:globalmovement} returns the seam to its starting bond.
Every embedded gate respects the charge selection rule, and translation
covariance expands the $L$th power of $\mathcal M_{\eta,\sigma}$ into the
ordered cycle.  No finite-order claim for the holonomy is required.
\end{proof}

\subsection{Relative-flux protection and the antiunitary}

We now give the complete antiunitary construction used in the theorem.  Let
$J|n\rangle=|-n\pmod N\rangle$, so
$JXJ=X^{-1}$ and $JZJ=Z^{-1}$, and let $\mathcal K_{z,a}$ be coefficientwise
complex conjugation in the physical $\sigma^z$ basis and the endpoint basis:
\begin{equation}
 \mathcal K_{z,a}
 \left(\sum_{\boldsymbol \xi,n}c_{\boldsymbol \xi,n}
 |\boldsymbol \xi;n\rangle\right)
 =\sum_{\boldsymbol \xi,n}c_{\boldsymbol \xi,n}^*
 |\boldsymbol \xi;n\rangle .
 \label{eq:Kza}
\end{equation}
For each $\eta$, choose the compatible modular lift
\begin{equation}
 \begin{aligned}
 M\kappa_\eta&\equiv-\eta\pmod N,\qquad 0\le\kappa_\eta<N,\\
 \tau_\eta&=-\frac{M\kappa_\eta+\eta}{N}\in\mathbb Z .
 \end{aligned}
 \label{eq:etakappatau}
\end{equation}
Put $\delta=N\bmod2$.  The diagonal endpoint shear $P_N$ used below is fixed
by
\begin{equation}
 P_NXP_N^\dagger=q^{1-\delta}XZ,\qquad
 P_NZP_N^\dagger=Z,\qquad P_NJP_N^*J=Z^{-\delta}.
 \label{eq:PNrelations}
\end{equation}
Its cyclic construction, including $N=2$, is given in
\ref{app:global}.  At the reference seam define
\begin{align}
 E_0&=X^{-L-1}P_N,\nonumber\\
 V_0&=\bigotimes_{j=0}^{L-1}
 q^{(2L-\delta-2j)\sigma_j^z/2},\nonumber\\
 C_0&=(V_0\otimes E_0)R,\nonumber\\
 \mathcal C_\eta
 &=\bigl((F_+^{(\eta)})^{\otimes L}\otimes I_a\bigr)C_0,\nonumber\\
 \mathcal A_\eta&=\mathcal C_\eta\Theta_\eta,\nonumber\\
 \Theta_\eta
 &=\bigl[(\sigma^x)^{\otimes L}\otimes
 X^{\kappa_\eta}J_a\bigr]\mathcal K_{z,a},\nonumber\\
 \mathcal B_{\eta,r}^{+}
 &=\mathcal M_{\eta,+}^{\,r}\mathcal A_\eta,
 \qquad r\in\mathbb Z,\nonumber\\
 \mathcal B_{\eta,r}^{-}
 &=\mathcal C_\eta^\dagger\!\mathcal B_{\eta,r}^{+}\!\mathcal C_\eta,
 \label{eq:etaoperators}
\end{align}
where $R O_jR=O_{L-1-j}$.  Thus the order of the antiunitary action and its
complex-conjugation basis are explicit.

\begin{theorem}[Charge-sector antiunitary involution]
\label{thm:sectorinvolution}
For every output label $\eta$, $\mathcal B_{\eta,r}^{\sigma}$ is a
same-branch antiunitary symmetry.  If the $Q$ sector is labelled by
$Q=\omega^m$, it is mapped to
\begin{equation}
 m\longmapsto c_\eta^\sigma-m\pmod N,
 \qquad
 c_\eta^+=2L+1-\kappa_\eta,\qquad
 c_\eta^-=-1-\kappa_\eta .
 \label{eq:etasectormap}
\end{equation}
Consequently every nonfixed pair of sectors has exactly the same energy
spectrum, including multiplicities.
\end{theorem}

\begin{proof}[Proof sketch]
The explicit branch maps obey
\begin{equation*}
 \mathcal C_\eta H_\eta^-\mathcal C_\eta^\dagger=H_\eta^+,
 \qquad
 \Theta_\eta H_\eta^+\Theta_\eta^{-1}=H_\eta^-.
\end{equation*}
Their composition with the co-moving translation gives the plus-branch antiunitary in
Eq.~\eqref{eq:etaoperators}.  Its charge action yields $c_\eta^+$.  Branch
exchange shifts the charge label by $L+1$, so conjugating the plus map by
$\mathcal C_\eta$ gives
$c_\eta^-=c_\eta^+-2(L+1)=-1-\kappa_\eta$ modulo $N$.  This relabelling is
essential when the two branches are compared sector by sector.
\end{proof}

The protected relative-holonomy orbit can be written equivalently as
\begin{equation}
 \begin{aligned}
 \nu_\eta=+1
 &\quad\Longleftrightarrow\quad
 \eta\equiv M-1\pmod2\\
 &\quad\Longleftrightarrow\quad
 f_\eta f_{\rm VMY}^{-1}\in\langle\eps^2\rangle=\mu_N,
 \qquad f_{\rm VMY}=\eps^{1-M}.
 \end{aligned}
 \label{eq:protectedphasecoset}
\end{equation}

\begin{corollary}[Even $N$: inter-sector pairing]
\label{cor:evenintersector}
If $N$ is even and $\nu_\eta=+1$, both $c_\eta^+$ and $c_\eta^-$ are odd.
The involution in Eq.~\eqref{eq:etasectormap} therefore has no fixed sector.
Every definite-$Q$ energy eigenstate has a linearly independent partner of
the same energy in a \emph{different} $Q$ sector.  This conclusion is
inter-sector spectral pairing, not a Kramers degeneracy inside each sector.
\end{corollary}

\begin{theorem}[Odd $N$: fixed-sector Kramers pairing]
\label{thm:oddkramers}
If $N$ is odd and $\nu_\eta=+1$, each branch has a unique fixed sector under
Eq.~\eqref{eq:etasectormap}.  With $r_0=(L+1)\bmod2$,
\begin{equation}
 \left.(\mathcal B_{\eta,r_0}^{\sigma})^{2}\right|_{m_*^\sigma}=-I .
 \label{eq:fixedsquare}
\end{equation}
Every eigenstate in that fixed sector therefore has a linearly independent
same-energy Kramers partner within the sector.
\end{theorem}

\begin{proof}[Proof sketch]
For odd $N$, multiplication by two is invertible modulo $N$, so the sector
involution has one fixed point.  On the plus branch, restriction of the exact
full-space square to that sector gives the exponent
\begin{equation}
 E_{\eta,r}\equiv(M+\eta)(L+r)\pmod2 .
 \label{eq:etafixedparity}
\end{equation}
The protected condition and $r=r_0$ make both factors odd.  Equation
\eqref{eq:fixedsquare} follows; the minus-branch statement is its unitary
image under $\mathcal C_\eta$.  Scalar phases are conjugated at every
antiunitary step.  The complete full-space square is given in
\ref{app:global}.
\end{proof}

\begin{corollary}[Full-space evenness in the protected orbit]
\label{thm:outputZ2}
\label{thm:evenmultiplicity}
For every coprime $(M,N)$, $L\ge3$, both branches, and every seam position,
\begin{equation}
 \begin{aligned}
 \nu_\eta=+1\quad\Longrightarrow\quad
 &\dim E_\lambda\!\left(H_\eta^\sigma(M;k)\right)\in2\mathbb Z\\
 &\text{for every eigenvalue }\lambda .
 \end{aligned}
 \label{eq:protectedcoset}
\end{equation}
For even $N$ this follows entirely from distinct-sector pairing.  For odd
$N$ the nonfixed sectors pair and the unique fixed sector is Kramers doubled.
Higher even multiplicities are allowed.
\end{corollary}

The source value $\eta=1-M$ is protected for every coprime root.  The
ungauged value $\eta=0$ is protected exactly when $M$ is odd.  In the opposite
holonomy orbit the same construction does not provide full-space evenness;
exact and finite-size spectral controls are given next.

\section{Exact physical controls}
\label{sec:controls}

The opposite output class provides a direct test of the distinction between
local and global equivalence.  At $(N,M,L)=(3,2,3)$, the ungauged choice
$\eta=0$ lies outside the protected class.  For either branch,
\begin{equation}
 \begin{aligned}
 \det(\lambda I-H_{\rm Weyl}^{\sigma})
 &=p_{\rm fixed}(\lambda)\,p_{\rm pair}(\lambda)^2,\\
 \deg p_{\rm fixed}&=\deg p_{\rm pair}=8,
 \end{aligned}
 \label{eq:exactcounterfactor}
\end{equation}
where $p_{\rm fixed}$ and $p_{\rm pair}$ are square-free and mutually
 coprime.  The spectrum therefore contains eight simple and eight double
 levels.  Exact polynomial coefficients are given in
\ref{app:controls}.

\begin{corollary}[Flux-orbit separation]
\label{cor:twofluxseparation}
At $(N,M,L)=(3,2,3)$, the ungauged representative has $\nu=-1$ and the VMY
representative has $\nu=+1$.  For arbitrary branches
$\sigma,\tau\in\{+,-\}$,
\begin{equation}
 H_{\rm Weyl}^{\sigma}\not\simeq_{\rm U}H_{\rm VMY}^{\tau}.
 \label{eq:twofluxseparation}
\end{equation}
The former has eight simple and eight double levels; the latter has twelve
double levels.
\end{corollary}

\begin{proof}[Proof sketch]
Exact arithmetic over $\mathbb Q(q)/(q^2+q+1)$ gives
\begin{equation}
 \chi_{\rm VMY}^{+}(\lambda)=\chi_{\rm VMY}^{-}(\lambda)=C(\lambda)^2,
 \label{eq:protectedfactor}
\end{equation}
where $C$ is square-free of degree $12$.  Together with
Eq.~\eqref{eq:exactcounterfactor}, the two multiplicity patterns are distinct,
and unitary conjugacy preserves them.  The exact factorizations and
coprimality calculations are reproduced in \ref{app:controls}.
\end{proof}

A sector-resolved scan on a fixed parameter grid tests one representative of
each holonomy orbit for $2\le N\le12$, every coprime $M$, $L=3,4,5$, and both
branches.  The preregistered canonical clustering tolerance was $10^{-9}$; the
final audit reports the stricter value $10^{-10}$ as its primary tolerance.
All $270$
protected Hamiltonians have even full-space multiplicities, whereas every one
of the $270$ unprotected Hamiltonians has at least one simple level.  Every
partition is unchanged from $10^{-8}$ through $10^{-11}$; the older
$10^{-7}$ value is retained only as a successful stress test.  The smallest
inter-cluster gap is
$6.06334\times10^{-7}$ at $(N,M,L,\sigma,\eta)=(10,1,5,-,0)$.
This tight case was checked using 80-digit working arithmetic, with a Jacobi
off-diagonal Frobenius residual below $6.6\times10^{-62}$; the full-precision
value and diagnostics are available with the supporting data.  Protected
cases have maximum multiplicity two in 268 instances
and four in two instances; unprotected cases have maximum multiplicity two in
268 instances and six in two instances.  Thus ``unprotected'' does not mean
nondegenerate.  Table~\ref{tab:scan-summary} summarizes these finite data.

\begin{table*}[t]
\centering
\caption{Sector-resolved finite scan at primary clustering tolerance
$10^{-10}$.   Maximum
multiplicity is a diagnostic, not part of the analytic theorem.}
\label{tab:scan-summary}
\begin{tabular}{p{0.15\textwidth}p{0.07\textwidth}p{0.54\textwidth}p{0.08\textwidth}}
\hline
Class & Cases & Sector-resolved result & Max. mult. \\
\hline
Protected, even $N$
& 102
& no fixed sector; full spectrum even; in $100/102$ cases every sector is internally simple
& 4 \\
Protected, odd $N$
& 168
& one internally even fixed sector; nonfixed sectors internally simple; full spectrum even
& 2 \\
Unprotected
& 270
& a full-space simple level observed in every scanned case
& 6 \\
\hline
\end{tabular}
\end{table*}

The same data expose the mechanism split: protected nonfixed sectors can
contain simple levels internally even though the full spectrum is paired,
while the unique fixed sector for odd $N$ is internally even.  A second exact
orbit separation occurs at $(N,M,L)=(2,1,4)$: the protected representative
has histogram $2{:}12;4{:}2$, whereas the opposite representative contains
eight simple levels.  Full scan specifications are given in
\ref{app:controls}; the case-resolved records are available with the
supporting data.

The transmutation theorem also permits a continuous twist.  At
$(N,M,L)=(3,2,3)$ define
$K_\sigma(t)=H_0^\sigma-h_{01}+h_{01}^{[t]}$ for
$t=\exp(i\phi)\in U(1)$.  At the six output-normalization points, the
spectra alternate between eight simple plus eight double levels when
$t^3=+1$ and twelve double levels when $t^3=-1$.  On a separate 3600-point
midpoint mesh avoiding those values, every sampled spectrum has 24 simple
levels.  This finite-grid observation does not analytically exclude isolated
nonspecial crossings.  The flow is shown in
Fig.~\ref{fig:continuous-flux}.

\begin{figure*}[t]
\centering
\includegraphics[width=\textwidth]{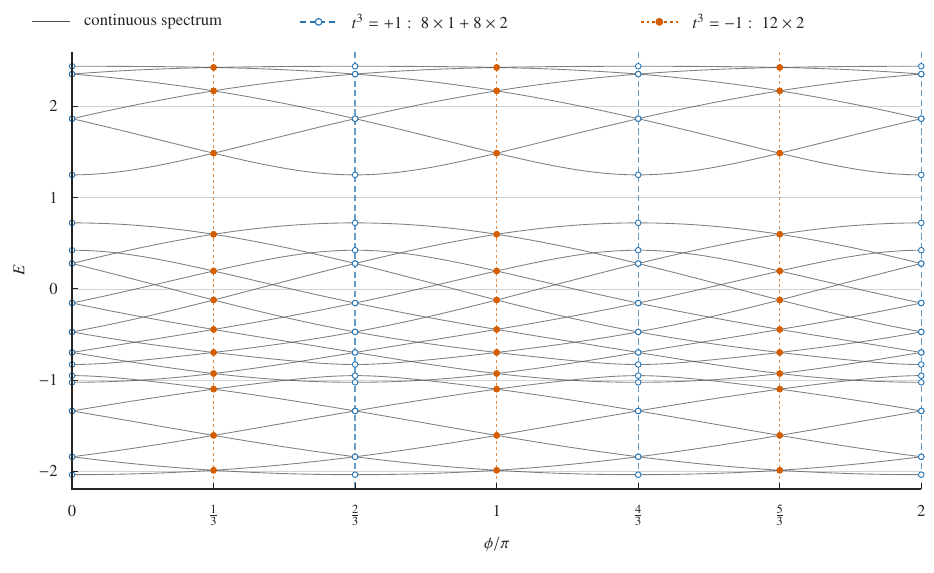}
\caption{Continuous relative-flux spectrum of
$K_+(t)=H_0^+-h_{01}+h_{01}^{[t]}$ at $(N,M,L)=(3,2,3)$.  Blue open markers
have $t^3=+1$ and multiplicity histogram $1{:}8;2{:}8$; orange filled
markers have $t^3=-1$ and histogram $2{:}12$.  The minus-branch spectrum
agrees on the plotted grid to $8.9\times10^{-15}$.  Curves between the six
exact output points are numerical.  Energies are shown in the Hamiltonian
normalization of Eq.~\eqref{eq:density}.}
\label{fig:continuous-flux}
\end{figure*}

\section{Discussion and limitations}
\label{sec:limitations}

Equation~\eqref{eq:fluxtransmutation} gives the natural global variable that
was hidden in the local output gauge.  Gauge transport to a bond twist is a
standard ring phenomenon.  What is model specific is the next step: the
relative holonomy selects an explicit charge-sector antiunitary, including its
sector involution and projective square.  Neither a generic Aharonov--Bohm
argument nor local seam mobility alone supplies those operator identities.

The even- and odd-$N$ conclusions should also be kept conceptually distinct.
For even $N$, the symmetry enforces equality of spectra in different charge
sectors; an individual sector can contain simple levels.  For odd $N$, the
unique fixed sector carries an additional square-minus-one antiunitary and is
Kramers doubled internally.  Full-space evenness is a useful common corollary,
but it is not the whole mechanism.  The multiplicity-four levels at
$(2,1,4)$ further show that pairing does not mean exact twofold degeneracy.

The retained endpoint is a dynamical $N$-state degree of freedom: it carries
the Weyl algebra, enters the combined charge, participates in the movement
gates, and implements the exact equivalence in Eq.~\eqref{eq:etaplustwo}.
Accordingly, $H_\eta^\sigma$ is an exactly characterized, exactly movable
impurity Hamiltonian with a proved symmetry classification; no closed-form
solution of its many-body spectrum is claimed.  Its operator domain is an
explicit root-of-unity XXZ realization of the broader impurity construction
in which an MPO virtual space remains as a dangling dynamical Hilbert factor
\citep{Ueda2026Perfect}.  Relating it to the endpoint-traced temporal operator
would require a temporal/spatial endpoint map and junction data.

Sequentially movable impurities and dangling virtual spaces have direct
precedents, most notably \citep{Ueda2026Perfect}; antiunitary pairing and
closed defect-motion cycles also occur in other operator settings
\citep{Liang2026Physical,Aasen2020Dualities,Sato2026Invariants}.  Root-of-unity
XXZ chains additionally host loop-algebra multiplets, twisted sectors,
Fabricius--McCoy strings, $Q$-operator descendants, and hidden sectors
\citep{DeguchiFabriciusMcCoy2001Loop,Deguchi2004Twisted,Korff2004Twisted,Miao2021QOperator,Hu2026Hidden}.
The familiar twist pattern is that generic twists lift sector-internal
degeneracies while special root-of-unity twists can restore enhanced
symmetries \citep{Deguchi2004Twisted,Korff2004Twisted}.  Here the movable seam
turns its discrete output freedom into two relative holonomy orbits and makes
the protected special orbit explicit at the operator level.  The present
mechanism combines this flux reduction with a retained endpoint and a
charge-sector antiunitary acting on every energy.  A source-by-source
comparison is provided in \ref{app:prior-boundaries}.

We do not derive this pairing from an affine Temperley--Lieb or loop-algebra
module decomposition.  Whether the retained-endpoint Hilbert space admits
such an organization remains an open problem.

\ref{app:controls} records two boundary diagnostics.  An exact
second moment separates the direct seam from a uniform chain tensored with an
endpoint identity.  In addition, one undressed cyclic transfer member does
not commute with the retained-endpoint seam.  This is consistent with VMY's
emphasis on noncommuting transfer structures and noncommuting conserved
quantities.  Constructing an endpoint-bearing or seam-dressed commuting
family remains open.

Other open problems include an endpoint lift, transport criteria, and a
thermodynamic or continuum interpretation, each requiring new analysis.

\section{Conclusion}

The output-normalized seam family has a simple global organization.  A local
gauge can be moved to one ordinary bond, and endpoint conjugation collapses
the $2N$ labels into two relative $\mathbb Z_2$ holonomy orbits.  The
independently derived condition $\eta\equiv M-1\pmod2$ selects the protected
orbit, so the source value $\eta_{\rm VMY}=1-M$ is protected at every coprime
root.  There the projective antiunitary
produces inter-sector spectral pairing for even $N$ and fixed-sector Kramers
pairing for odd $N$; together these imply even-dimensional energy eigenspaces.

The local theorem is more general than this finite cyclic conclusion: any
unitary $q$-Weyl pair gives two-sided gates, explicit Hermitian seams, exact
movement, and a universal quartic in the unrestricted Laurent--Weyl algebra.
At finite roots, the exact unprotected example and the broader scan show that
the holonomy distinction has spectral content rather than being a redundant
label.  The analytic converse remains open, as does the separate problem of
relating this retained-endpoint Hamiltonian to the endpoint-traced temporal
defect.

The technical material is organized as follows.  \ref{app:unitarity}
establishes the Weyl-algebra statements; \ref{app:vmyprovenance} records the
source map and flux transport; \ref{app:ansatz} gives the seam construction
and local spectrum; \ref{app:uniqueness} proves support-qualified uniqueness;
\ref{app:global} derives the projective antiunitary; and \ref{app:controls}
collects exact controls, diagnostics, and reproduction conventions.

\section*{Data availability}

The data and code used to reproduce the computational results and figures
reported in this article are publicly available on Zenodo \cite{Yu2026XXZSeamData} at
\href{https://doi.org/10.5281/zenodo.22069216}{doi:10.5281/zenodo.22069216}.

% \section*{CRediT authorship contribution statement}
% [Author contributions to be confirmed by all authors before submission.]

% \section*{Declaration of competing interest}
% [Competing-interest declaration to be confirmed by all authors before submission.]

\section*{Acknowledgements}
This work was supported in part by the National Natural Science Foundation of
China (NSFC) under Grant No. 12347147.

\appendix
\section{Weyl algebra and finite cyclic specialization}
\label{app:unitarity}

Let $q\in U(1)$ and let $X,Z$ be bounded unitaries on an endpoint Hilbert space $\cH_a$ satisfying
\begin{equation}
\omega=q^2,\qquad ZX=\omega XZ.
\end{equation}
Because $X$ and $Z$ are unitary, $X^{-1}=X^\dagger$ and
$Z^{-1}=Z^\dagger$.  We use inverse notation in the defining blocks and
dagger notation when displaying adjoint products.  No root quotient, finite
dimension, or irreducibility is assumed.  The identities needed in the four
block products are
\begin{align}
Z^\dagger XZ&=\omega^{-1}X,&
Z^\dagger X^\dagger Z&=\omega X^\dagger,\nonumber\\
Z^\dagger X^\dagger&=\omega X^\dagger Z^\dagger,&
ZX^\dagger&=\omega^{-1}X^\dagger Z,&
XZ^\dagger&=\omega Z^\dagger X.
\label{eq:suppweylids}
\end{align}
All follow from the single Weyl relation, with no finite-root interpolation.

\begin{samepage}
For the plus branch, take
\begin{equation}
L_+=\begin{pmatrix}I&\omega X^{-1}\\-qXZ&qZ\end{pmatrix},
\end{equation}
\end{samepage}
the four blocks of $L_+^\dagger L_+$ are
\begin{equation}
\begin{pmatrix}
I+Z^\dagger X^\dagger XZ & \omega X^\dagger-Z^\dagger X^\dagger Z\\
\omega^{-1}X-Z^\dagger XZ & XX^\dagger+Z^\dagger Z
\end{pmatrix}
=\begin{pmatrix}2I&0\\0&2I\end{pmatrix}.
\label{eq:suppLpdagLp}
\end{equation}
The reverse product is
\begin{equation}
\begin{aligned}
L_+L_+^\dagger
&=\begin{pmatrix}
I+X^\dagger X & q^{-1}(-Z^\dagger X^\dagger+\omega X^\dagger Z^\dagger)\\
q(-XZ+\omega^{-1}ZX) & XX^\dagger+ZZ^\dagger
\end{pmatrix}\\
&=\begin{pmatrix}2I&0\\0&2I\end{pmatrix}.
\end{aligned}
\label{eq:suppLpLpdag}
\end{equation}

For the minus branch, take
\begin{equation}
L_-=\begin{pmatrix}qZ&-qX^{-1}Z\\X&\omega I\end{pmatrix},
\end{equation}
the first order gives
\begin{equation}
L_-^\dagger L_-=
\begin{pmatrix}
Z^\dagger Z+X^\dagger X & -Z^\dagger X^\dagger Z+\omega X^\dagger\\
-Z^\dagger XZ+\omega^{-1}X & Z^\dagger XX^\dagger Z+I
\end{pmatrix}
=\begin{pmatrix}2I&0\\0&2I\end{pmatrix},
\label{eq:suppLmdagLm}
\end{equation}
and the reverse order is
\begin{equation}
L_-L_-^\dagger=
\begin{pmatrix}
ZZ^\dagger+X^\dagger ZZ^\dagger X & q(ZX^\dagger-\omega^{-1}X^\dagger Z)\\
q^{-1}(XZ^\dagger-\omega Z^\dagger X) & XX^\dagger+I
\end{pmatrix}
=\begin{pmatrix}2I&0\\0&2I\end{pmatrix}.
\label{eq:suppLmLmdag}
\end{equation}
Equations~\eqref{eq:suppLpdagLp}--\eqref{eq:suppLmLmdag} establish both multiplication orders separately.  The normalized gates $W_\pm=L_\pm/\sqrt2$ are therefore ordinary unitaries in every unitary $q$-Weyl representation.

\subsection{Finite dimension forces a root}

If $X,Z$ are invertible $d\times d$ matrices, taking determinants of $ZX=q^2XZ$ gives
\begin{equation}
 \det Z\,\det X=q^{2d}\det X\,\det Z.
\end{equation}
Thus
\begin{equation}
 (q^2)^d=1.
 \label{eq:suppdetroot}
\end{equation}
If $q^2$ has primitive order $N$, every finite representation decomposes into orbit subspaces whose dimensions are multiples of $N$.  In an irreducible unitary representation, an eigenvector $Z|0\rangle=z_0|0\rangle$ generates
$|n\rangle=X^n|0\rangle$ with pairwise distinct $Z$ eigenvalues $z_0q^{2n}$ for $n=0,\ldots,N-1$; their span is invariant under both generators.  Irreducibility makes it the whole space.  After scalar rephasings of $X$ and $Z$ and a diagonal basis rephasing, this is the standard $N$-dimensional clock--shift pair.  Reducible finite representations need not have scalar Weyl commutant.

\section{VMY source map and exact flux transport}
\label{app:vmyprovenance}

VMY begin with the cyclic Lax tensor \citep{Vernier2026Lattice}
\begin{equation}
\mathcal L(u,v,s)=
\begin{pmatrix}
\ee^{s-u}+\ee^{u-s}Z
&
\omega(\ee^{s-v}-\ee^{v-s}Z)X^{-1}
\\
X(\ee^{v+s}-\ee^{-v-s}Z)
&
\omega\ee^{u+s}+\ee^{-u-s}Z
\end{pmatrix}.
\label{eq:suppvmyLax}
\end{equation}
Their endpoint-traced duality operators arise from the $s=0$ singular limits;
the limiting physical-row factors are
\begin{equation}
 \eps^{L/2-S^z}=\bigotimes_j\operatorname{diag}(1,\eps)_j,
 \qquad
 \eps^{L/2+S^z}=\bigotimes_j\operatorname{diag}(\eps,1)_j.
\label{eq:supprowfactors}
\end{equation}
In our physical basis the resulting local blocks are
\begin{equation}
\mathcal L^+_{\rm VMY}=
\begin{pmatrix}I&\omega X^{-1}\\-\eps XZ&\eps Z\end{pmatrix},
\qquad
\mathcal L^-_{\rm VMY}=
\begin{pmatrix}\eps Z&-\eps X^{-1}Z\\X&\omega I\end{pmatrix},
\label{eq:suppvmyBlocks}
\end{equation}
where the finite cyclic parameters are
\begin{equation}
 \begin{aligned}
 \eps&=\ee^{\ii\pi/N},\qquad q=\eps^M,\qquad \omega=q^2,\\
 &1\leq M<N,\qquad \gcd(M,N)=1.
 \end{aligned}
 \label{eq:suppgenericroot}
\end{equation}
In the printed v2 preprint, the local duality blocks occur in
Eqs.~(SD.5)--(SD.6), while $q=\eps^M$ is introduced separately in
Eq.~(SD.16).  The physical rows contain the literal principal phase $\eps$;
this distinction is part of the source formula.
The principal phase $\eps$ in the physical rows and the Weyl phase $\omega$
must both be retained.  Put
\begin{equation}
 f=\frac{\eps}{q},\qquad
 F_+=\operatorname{diag}(1,f),\qquad
 F_-=\operatorname{diag}(f,1).
 \label{eq:suppFpm}
\end{equation}
Direct multiplication of the output rows gives the exact branch-dependent
source map
\begin{equation}
 \begin{aligned}
 \mathcal L^+_{\rm VMY}&=F_+L_+(q),\qquad
 \mathcal L^-_{\rm VMY}=F_-L_-(q),\\
 W_\sigma^{\rm VMY}&=F_\sigma W_\sigma.
 \end{aligned}
 \label{eq:suppsourcegauge}
\end{equation}
Here and below $\sigma\in\{+,-\}$ labels the seam branch.  These are left factors, not similarities of the local gates.  Since
$F_\sigma$ is unitary, Eq.~\eqref{eq:suppsourcegauge} and
Eqs.~\eqref{eq:suppLpdagLp}--\eqref{eq:suppLmLmdag} prove both ordinary
unitarity products for every coprime $(M,N)$.

The bulk density transports an arbitrary diagonal unitary from one end of
an oriented bond to the other:
\begin{equation}
 F_B^\dagger h_{BC}(q)F_B=F_Ch_{BC}(q)F_C^\dagger.
 \label{eq:suppbulktransport}
\end{equation}
Indeed, if $F=\operatorname{diag}(d_0,d_1)$ and
$\chi=d_0^*d_1$, the two sides multiply
$\sigma_B^+\sigma_C^-$ by the same $\chi$, its adjoint by $\chi^*$,
and leave $\sigma_B^z\sigma_C^z$ fixed.  If $S_\sigma(q)$ denotes the
twelve-term unrestricted seam constructed below, define
\begin{equation}
 S_{\sigma,ABa}^{\rm VMY}(M)=F_{\sigma,B}S_{\sigma,ABa}(q)F_{\sigma,B}^\dagger.
 \label{eq:suppSVMY}
\end{equation}
The unrestricted movement identity and Eq.~\eqref{eq:suppbulktransport}
then give, without interpolation in $M$,
\begin{equation}
 (W_{\sigma,Ba}^{\rm VMY})^\dagger
 \bigl[h_{BC}+S_{\sigma,ABa}^{\rm VMY}\bigr]W_{\sigma,Ba}^{\rm VMY}
 =h_{AB}+S_{\sigma,BCa}^{\rm VMY}.
 \label{eq:suppVMYmovement}
\end{equation}

For clarity, the finite discussion contains three different objects:
\begin{align}
 H_{\rm Weyl}^\sigma(q;k)
 &=\sum_{j\neq k-1}h_{j,j+1}(q)+S_{\sigma,k-1,k,a}(q),
 \label{eq:suppHWeyl}\\
 H_{\rm VMY}^\sigma(M;k)
 &=\sum_{j\neq k-1}h_{j,j+1}(q)+S_{\sigma,k-1,k,a}^{\rm VMY}(M),
 \label{eq:suppHVMY}
\end{align}
More generally, for $\eta\in\mathbb Z_{2N}$ set
\begin{align}
 f_\eta&=\eps^\eta,\nonumber\\[-0.4ex]
 F_+^{(\eta)}&=\operatorname{diag}(1,f_\eta),&
 F_-^{(\eta)}&=\operatorname{diag}(f_\eta,1),
 \label{eq:suppoutputgaugefamily}\\
 W_\sigma^{(\eta)}&=F_\sigma^{(\eta)}W_\sigma,&
 S_\sigma^{(\eta)}&=F_{\sigma,B}^{(\eta)}S_\sigma
 F_{\sigma,B}^{(\eta)\dagger},\nonumber\\
 H_\eta^\sigma(M;k)&=\sum_{r\neq k-1}h_{r,r+1}(q)
                   +S_{\sigma,k-1,k,a}^{(\eta)}.
 \label{eq:suppHeta}
\end{align}
Then $H_0^\sigma=H_{\rm Weyl}^\sigma$ and
$H_{1-M}^\sigma=H_{\rm VMY}^\sigma$, with $\eta$ modulo $2N$.  The diagonal bond
transport proves the same local unitarity, movement, spectrum, and affine
kernel statements for every member.  \ref{app:output-z2} will show
how the resulting relative $\mathbb Z_2$ holonomy controls the global
antiunitary protection.

\subsection{Exact flux transmutation and output-orbit reduction}

Put $\Gamma(t)=\operatorname{diag}(1,t)$ and
$\chi_+=1$, $\chi_-=-1$.  Up to a scalar
that drops out of conjugation,
$F_\sigma^{(\eta)}=\Gamma(f_\eta^{\chi_\sigma})$.  For $|t|=1$ define
\begin{align}
 h_{j,j+1}^{[t]}
 &:=\Gamma_j(t)^\dagger h_{j,j+1}\Gamma_j(t)
 =\Gamma_{j+1}(t)h_{j,j+1}\Gamma_{j+1}(t)^\dagger,
 \label{eq:supptwistedbonddef}\\
 &=qt\,\sigma_j^+\sigma_{j+1}^-
 +q^{-1}t^{-1}\,\sigma_j^-\sigma_{j+1}^+
 +\frac{q+q^{-1}}4\sigma_j^z\sigma_{j+1}^z .
 \label{eq:supptwistedbond}
\end{align}
Choose a reference exponent $\eta_0$ and set
\begin{equation}
 r_{\eta|\eta_0}=\frac{f_\eta}{f_{\eta_0}},\qquad
 t_{\sigma,\eta|\eta_0}=r_{\eta|\eta_0}^{\chi_\sigma}.
 \label{eq:supprelativetwist}
\end{equation}
If the seam occupies $(k-1,k)$, direct one-site conjugation gives
\begin{equation}
 \Gamma_k(t)^\dagger H_\eta^\sigma(M;k)\Gamma_k(t)
 =H_{\eta_0}^\sigma(M;k)-h_{k,k+1}+h_{k,k+1}^{[t]},
 \qquad t=t_{\sigma,\eta|\eta_0}.
 \label{eq:suppfluxtransmutation}
\end{equation}
Indeed the conjugation changes the seam gauge from
$\Gamma(f_\eta^{\chi_\sigma})$ to
$\Gamma(f_{\eta_0}^{\chi_\sigma})$ and acts on no ordinary
density except $h_{k,k+1}$.  For $p=0,\ldots,L-2$ define
\begin{equation}
 U_p(t)=\prod_{\ell=0}^{p}\Gamma_{k+\ell}(t)^\dagger,\qquad
 b_p=k+p\pmod L .
 \label{eq:suppfluxstring}
\end{equation}
Induction with Eq.~\eqref{eq:supptwistedbonddef} then yields
\begin{equation}
 U_p(t)H_\eta^\sigma(M;k)U_p(t)^\dagger
 =H_{\eta_0}^\sigma(M;k)-h_{b_p,b_p+1}
  +h_{b_p,b_p+1}^{[t]}.
 \label{eq:suppfluxatbond}
\end{equation}
Thus the relative output gauge is exactly one directed twist, and its
location on the ordinary part of the ring is a gauge choice.

There is a second exact equivalence.  Let
$m_{\rm inv}\in\{0,\ldots,N-1\}$ be the unique residue satisfying
\begin{equation}
 Mm_{\rm inv}\equiv1\pmod N.
 \label{eq:suppminv}
\end{equation}
For a seam monomial $E_{rc}X^aZ^b$, let
$n_B(r):=r\bmod2$ denote the down-spin occupation on the second physical
leg.  Inspection of the twelve nonzero terms gives the branch-uniform
selection rule
\begin{equation}
 b=\chi_\sigma\bigl(n_B(r)-n_B(c)\bigr).
 \label{eq:suppoutputselection}
\end{equation}
Writing $\zeta=\eps^2$ and using $q^2=\zeta^M$, endpoint conjugation gives
\begin{equation}
 X^{-m_{\rm inv}}(X^aZ^b)X^{m_{\rm inv}}
 =(q^2)^{m_{\rm inv}b}X^aZ^b=\zeta^bX^aZ^b.
 \label{eq:suppendpointphase}
\end{equation}
This is precisely the output-row phase generated by
$f_{\eta+2}=\zeta f_\eta$, and hence
\begin{equation}
 H_{\eta+2}^{\sigma}(M;k)
 =X_a^{-m_{\rm inv}}H_\eta^\sigma(M;k)X_a^{m_{\rm inv}}
 \label{eq:suppetaplustwo}
\end{equation}
for both branches and every seam position.  Each parity coset of
$\eta\in\mathbb Z_{2N}$ is therefore one endpoint-Weyl-conjugacy orbit.  With
$Q=g^{\otimes L}\otimes Z_a^{-1}$, the same endpoint conjugation satisfies
\begin{equation}
 X_a^{-m_{\rm inv}}QX_a^{m_{\rm inv}}
 =\omega^{-m_{\rm inv}}Q .
 \label{eq:suppchargeshift}
\end{equation}
Therefore the forward unitary identification from $H_\eta^\sigma$ to
$H_{\eta+2}^\sigma$ sends charge label $m$ to $m+m_{\rm inv}$ modulo $N$.
In the modular lift used in \ref{app:global}, this is equivalent to
\begin{equation}
 \kappa_{\eta+2}=\kappa_\eta-2m_{\rm inv},\qquad
 c_{\eta+2}^\sigma=c_\eta^\sigma+2m_{\rm inv}\pmod N.
 \label{eq:suppchargeshiftconsistency}
\end{equation}
With
$\eta_{\rm VMY}=1-M$, the quotient character is
\begin{equation}
 \nu_\eta=
 \left(\frac{f_\eta}{f_{\rm VMY}}\right)^N
 =(-1)^{\eta-(1-M)}\in\{+1,-1\},
 \qquad t_{\sigma,\eta|\eta_{\rm VMY}}^N=\nu_\eta .
 \label{eq:supprelativeflux}
\end{equation}
Thus $\mu_{2N}/\mu_N\cong\mathbb Z_2$ supplies exactly two relative-holonomy
labels.  Equation~\eqref{eq:suppetaplustwo} implies at most two classes under
arbitrary global unitary equivalence; an all-parameter proof that the two
orbits are always inequivalent is not assumed.

The notation $H_{\rm VMY}^\sigma$ denotes the retained-endpoint member
constructed from the VMY source rows with the fixed reshape and traceless
convention.  It acts on
$(\mathbb C^2)^{\otimes L}\otimes\mathbb C_a^N$, whereas
$D_\pm=\Tr_a\prod_j\mathcal L_{a,j}^\pm$ acts on
$(\mathbb C^2)^{\otimes L}$.  Equation~\eqref{eq:suppSVMY} is a local
one-leg conjugation; on the ring its residual action is the bond twist in
Eq.~\eqref{eq:suppfluxtransmutation}.

VMY prove transfer-matrix conservation for their finite cyclic,
endpoint-traced representation \citep{Vernier2026Lattice}.  The unrestricted $q$-Weyl seam
theorem is a local Laurent--Weyl identity and
does not presuppose that a trace-class or otherwise well-defined commuting
transfer matrix exists in an arbitrary representation.

There is also an exact, distinct degeneration that produces the generalized Weyl family.  Put $u=v=t$ and choose $\ee^{-2s}=q$.  Directly from Eq.~\eqref{eq:suppvmyLax},
\begin{equation}
\lim_{t\to-\infty}\ee^{t-s}\mathcal L(t,t,s)=L_+(q),
\qquad
\lim_{t\to+\infty}\ee^{-t-s}\mathcal L(t,t,s)=L_-(q).
\label{eq:suppnonzerolimits}
\end{equation}
In the finite cyclic representation, its endpoint-traced normalization lies
in the conserved VMY family, but it is not the generic-$M$ duality definition.
In a general
$q$-Weyl representation Eq.~\eqref{eq:suppnonzerolimits} identifies only
local algebraic blocks.  We make no direct-seam Onsager-exchange,
Tambara--Yamagami fusion, endpoint-trace, or transfer-family inference from
this limit.

\section{Construction and local spectrum of the movable seam}
\label{app:ansatz}

\subsection{Complete finite-dimensional space}

In the standard irreducible $N$-dimensional clock--shift realization, use the ordered two-spin basis $|0\rangle=|00\rangle$, $|1\rangle=|01\rangle$, $|2\rangle=|10\rangle$, $|3\rangle=|11\rangle$.  We write $Y_{m,n}=X^mZ^n$.  The operators
\begin{equation}
E_{rc}^{AB}\otimes Y_{m,n},\qquad r,c=0,1,2,3,\qquad m,n=0,\ldots,N-1,
\end{equation}
form a basis of the full algebra on $\mathbb C^2_A\otimes\mathbb C^2_B\otimes\mathbb C^N_a$.  The complete fixed-$N$ ansatz therefore has $16N^2$ complex coefficients.  Hermiticity, tracelessness, and sparsity are not assumed.

The endpoint exponents in the four $L_+$ blocks are
\begin{equation}
\mathcal G_+=\{(0,0),(-1,0),(1,1),(0,1)\},
\end{equation}
and those in $L_-$ are
\begin{equation}
\mathcal G_-=\{(0,1),(-1,1),(1,0),(0,0)\}.
\end{equation}
One conjugation generates differences in the rectangle $-2\leq m\leq2$, $-1\leq n\leq1$.  The universal source-generated lift is thus
\begin{equation}
\widetilde S=\sum_{r,c=0}^{3}\sum_{m=-2}^{2}\sum_{n=-1}^{1}
c_{rc}(m,n)E_{rc}\Y{m}{n},
\label{eq:supplift}
\end{equation}
with 240 independent variables.  All 240 appear in the generated equations.  This Laurent rectangle is a source-generated subspace used to construct a particular solution uniformly in $N$; it is not the complete $16N^2$-variable space at fixed $N$.  Completeness of the final classification comes instead from the full-algebra homogeneous-kernel proof in \ref{app:uniqueness}.

\subsection{Normal-ordering map}

Let $g_{rs}=\alpha_{rs}\Y{u_{rs}}{v_{rs}}$ and $g_{\ell w}=\beta_{\ell w}\Y{c_{\ell w}}{d_{\ell w}}$ be gate blocks.  The physical matrix-unit product is
$E_{sr}^{B}E_{j\ell}^{B}E_{\ell w}^{B}=\delta_{rj}E_{sw}^{B}$.  A generic conjugated ansatz term, including its zero value for $r\neq j$, is therefore
\begin{align}
&\frac12(E_{rs}^{B}g_{rs})^\dagger
(E_{j\ell}^{B}\Y{m}{n})(E_{\ell w}^{B}g_{\ell w})\nonumber\\
&\quad=\frac12\delta_{rj}\alpha_{rs}^{*}\beta_{\ell w}
q^{2u_{rs}v_{rs}-2v_{rs}m+2(n-v_{rs})c_{\ell w}}
E_{sw}^{B}\Y{m-u_{rs}+c_{\ell w}}{n-v_{rs}+d_{\ell w}}.
\label{eq:suppcoeffmap}
\end{align}
This formula and the matrix-unit expansion of the six nonzero entries of $h$ generate the coefficient of each
$E_{ik}^{A}E_{j\ell}^{B}E_{\mu\nu}^{C}\Y{m}{n}$ in the movement residual.

\subsection{Representative coefficient equations}

Normal-ordering every term with Eq.~\eqref{eq:suppcoeffmap} produces a finite linear system over $\mathbb C[q,q^{-1}]$.  All 240 coefficients in Eq.~\eqref{eq:supplift} occur.  The following rows illustrate the isolated and coupled parts of that system.

For the plus sign, three representative single-variable rows are
\begin{equation}
\begin{aligned}
-\frac12-c_{01}(-1,-1)&=0,
&-\frac12q^4-c_{03}(-2,-1)&=0,\\
\frac12q^2-c_{23}(-1,-1)&=0.&&
\end{aligned}
\end{equation}
The coupled rows for $c_{02}$ and $c_{13}$ reduce to
\begin{equation}
c_{02}+q^2c_{13}=0,\qquad q^{-2}c_{02}-c_{13}=q,
\end{equation}
whose solution is $c_{02}=q^3/2$, $c_{13}=-q/2$.  The paired lower entries follow from
\begin{equation}
c_{20}+q^{-2}c_{31}=0,\qquad
q^{-2}c_{20}-q^{-4}c_{31}=q^{-5},
\end{equation}
giving $c_{20}=q^{-3}/2$, $c_{31}=-q^{-1}/2$.  The accompanying exact coefficient data contain the full system in machine-readable form.

\subsection{Closed coefficients}

The nonzero coefficient tuples $(r,c,m,n,2c_{rc}(m,n))$ are
\begin{align}
\mathcal C_+=\{& (0,1,-1,-1,-1),(0,2,-1,0,q^3),(0,3,-2,-1,-q^4),\nonumber\\
&(1,0,1,1,-q^2),(1,2,0,1,q^2),(1,3,-1,0,-q),\nonumber\\
&(2,0,1,0,q^{-3}),(2,1,0,-1,q^{-2}),(2,3,-1,-1,q^2),\nonumber\\
&(3,0,2,1,-1),(3,1,1,0,-q^{-1}),(3,2,1,1,1)\},
\label{eq:Cplus}
\end{align}
and
\begin{align}
\mathcal C_-=\{& (0,1,-1,1,1),(0,2,-1,0,-q),(0,3,-2,1,-1),\nonumber\\
&(1,0,1,-1,q^{-2}),(1,2,0,-1,q^2),(1,3,-1,0,q^3),\nonumber\\
&(2,0,1,0,-q^{-1}),(2,1,0,1,q^{-2}),(2,3,-1,1,-q^{-2}),\nonumber\\
&(3,0,2,-1,-q^{-4}),(3,1,1,0,q^{-3}),(3,2,1,-1,-1)\}.
\label{eq:Cminus}
\end{align}
The factor $1/2$ is extracted in these lists.  Expanding the matrix units into Pauli operators gives 32 generic Pauli--Weyl coefficients per branch.

\subsection{Generic-\texorpdfstring{$M$}{M} source-normalized coefficients}

Equation~\eqref{eq:suppSVMY} gives an explicit coefficient rule without
re-solving the 240-variable system.  Let
\begin{equation}
 n_B(0),n_B(1),n_B(2),n_B(3)=(0,1,0,1),\qquad
 \chi_+=1,\quad \chi_-=-1,
 \label{eq:suppnBchi}
\end{equation}
where $n_B(r)$ is the down-spin number on the second physical leg of the
ordered two-spin basis.  If one row of $2S_\sigma(q)$ is
\begin{equation}
 \varsigma_{rc}q^{p_{rc}}E_{rc}Y_{m,n},\qquad
 \varsigma_{rc}\in\{+1,-1\},
\end{equation}
then the corresponding row of $2S_\sigma^{\rm VMY}(M)$ is
\begin{align}
 &\varsigma_{rc}q^{p_{rc}}
 f^{\chi_\sigma[n_B(r)-n_B(c)]}E_{rc}Y_{m,n}\nonumber\\
 &\qquad=\varsigma_{rc}
 q^{p_{rc}-\chi_\sigma[n_B(r)-n_B(c)]}
 \eps^{\chi_\sigma[n_B(r)-n_B(c)]}E_{rc}Y_{m,n}.
 \label{eq:suppVMYcoefficientrule}
\end{align}
No coefficient vanishes and no matrix-unit or Weyl exponent changes, so the
source-normalized seam has exactly the same twelve formal terms.  Unitary
conjugation proves Hermiticity and trace preservation.  The gauge map is a
linear bijection of the stated movement solution spaces, and
\begin{equation}
 [F_\sigma W_\sigma,I_2\otimes r]=F_\sigma[W_\sigma,I_2\otimes r].
 \label{eq:suppgaugecommutant}
\end{equation}
Consequently the full support-class kernel and affine family proved in
\ref{app:uniqueness} transfer unchanged.  This coefficient rule is a
local statement: because $F_+\neq F_-$ and only the second seam leg is
conjugated, it supplies neither a common full-chain conjugacy nor the global
branch and antiunitary formulas, which are derived separately in
\ref{app:global}.

\subsection{Complete movement-residual organization}
\label{app:movement}

For
\begin{equation}
\begin{aligned}
R_\sigma={}&W_{\sigma,Ba}^\dagger h_{BC}W_{\sigma,Ba}
+W_{\sigma,Ba}^\dagger S_{\sigma,ABa}W_{\sigma,Ba}\\
&-h_{AB}-S_{\sigma,BCa},
\end{aligned}
\end{equation}
substitution of Eqs.~\eqref{eq:Cplus}--\eqref{eq:Cminus}, followed by the normal-ordering rule in Eq.~\eqref{eq:suppcoeffmap}, leaves 36 occupied physical--Laurent basis elements for each branch before cancellation.  Collecting separately the contributions from $W^\dagger hW$, $W^\dagger SW$, $-h$, and $-S$ gives zero for every basis element.  The cancellations use only $ZX=q^2XZ$ and unitarity of $X$ and $Z$; no root quotient or numerical tolerance enters.  Together with the representative cancellations displayed in the main text, this establishes $R_\sigma=0$ in the unrestricted Laurent--Weyl algebra.  A machine-readable decomposition of the four contributions accompanies the reproducibility package.

For each integer $N\geq2$, evaluation under the quotient $X^N=Z^N=I_N$ and then in the standard $N$-dimensional Weyl representation preserves the zero residual.  Thus the identity holds in the stated endpoint representation, including the small-$N$ cases.

\subsection{Small-N alias handling}
\label{app:aliases}

The matrix-unit formula always has twelve distinct physical keys.  In the Pauli--Weyl expansion, endpoint aliases can cause coefficients attached to the same physical Pauli pair to combine.  Table~\ref{tab:aliases} reports aggregation after reducing endpoint exponents modulo $N$.

\begin{table*}[t]
\centering
\caption{Alias counts are identical for the two signs.  \emph{Zero} counts reduced Pauli--Weyl keys whose aggregated coefficient vanishes.}
\label{tab:aliases}
\begin{tabular}{rrrrrr}
\hline
$N$ & generic keys & reduced keys & aliases & zeros & nonzero reduced\\
\hline
2 & 32 & 12 & 20 & 9 & 3\\
\cline{1-6}
3 & 32 & 32 & 0 & 0 & 32\\
4 & 32 & 32 & 0 & 0 & 32\\
5 & 32 & 32 & 0 & 0 & 32\\
6 & 32 & 32 & 0 & 0 & 32\\
\hline
\end{tabular}
\end{table*}

The $N=2$ collapse is handled by aggregation, never by treating coincident Weyl monomials as independent.  It does not alter the formal movement proof.

\subsection{Universal quartic identity and finite local spectrum}
\label{app:localspectrum}

\subsubsection{Unrestricted Laurent--Weyl polynomial}
\label{app:localquartic}

For either branch, exact normal ordering of the closed twelve-term seam gives
\begin{equation}
 S_\sigma^4-\frac32S_\sigma^2+\frac{q+q^{-1}}2S_\sigma
 +\left[\frac1{16}-\frac{(q+q^{-1})^2}{16}\right]I=0 .
 \label{eq:suppquartic}
\end{equation}
This is an identity in the unrestricted Laurent--Weyl algebra.  Expanding each power before imposing any root quotient and collecting the sixteen occupied physical--Laurent basis elements gives zero coefficient by coefficient.  The accompanying exact data retain the separate $S^4$, $S^2$, $S$, and identity contributions as integer Laurent polynomials.
For the finite source-normalized seam,
$S_\sigma^{\rm VMY}=F_{\sigma,B}S_\sigma F_{\sigma,B}^\dagger$ gives the same quartic by
unitary similarity; this transfer does not assert an unrestricted VMY
object because $F_\sigma$ uses the finite phases $\eps$ and $q=\eps^M$.

\subsubsection{Finite charge blocks and local spectrum}
\label{app:localblocks}

Now let $q^2=\omega$ have primitive order $N$ and use the irreducible clock--shift endpoint.  Define
\begin{equation}
 g=\operatorname{diag}(1,\omega),\qquad
 Q_{\rm loc}=g\otimes g\otimes Z^{-1}.
 \label{eq:supplocalcharge}
\end{equation}
Every matrix-unit--Weyl monomial $E_{rc}Y_{m,n}$ in either seam satisfies the selection rule
\begin{equation}
 m=n_\downarrow(r)-n_\downarrow(c),
 \label{eq:suppselection}
\end{equation}
and therefore commutes with $Q_{\rm loc}$.  If $d(r)=n_\downarrow(r)$ for
$r=00,01,10,11$, the eigenvalue-$\omega^t$ sector has the basis
\begin{equation}
 |r;t\rangle=|r\rangle\otimes|d(r)-t\rangle,
 \qquad t\in\mathbb Z_N.
 \label{eq:supplocalbasis}
\end{equation}
Each sector is four dimensional.  For the plus seam the $t$ dependence is removed by
\begin{equation}
 D_t^+=\operatorname{diag}(1,q^{-2t},1,q^{-2t}).
\end{equation}
For the minus seam in the same $Q_{\rm loc}$ labeling it is removed by
\begin{equation}
 D_t^-=\operatorname{diag}(1,q^{2t},1,q^{2t}),
\end{equation}
followed by a fixed, $t$-independent gauge.  Hence the $N$ charge blocks of a fixed branch are unitarily equivalent.

To compute the reference blocks symmetrically, one may equivalently label the minus sectors by
$Q_{{\rm loc},-}=g_-\otimes g_-\otimes Z$, where
$g_-=\operatorname{diag}(\omega,1)$.  Since
$Q_{{\rm loc},-}=\omega^2Q_{\rm loc}^{-1}$, this only permutes the sectors of
the same cyclic symmetry.  Use
\begin{equation}
 |r;t\rangle_+=|r\rangle\otimes|d(r)-t\rangle,\qquad
 |r;t\rangle_-=|r\rangle\otimes|t-2+d(r)\rangle.
\label{eq:suppbranchbases}
\end{equation}
In these branch-adapted bases both families are gauged by
\begin{equation}
 D_t=\operatorname{diag}(1,q^{-2t},1,q^{-2t}).
\end{equation}
Direct substitution at $t=0$ gives
\begin{equation}
2B_0^+=
\begin{pmatrix}
0&-q^{-2}&q^3&-1\\
-q^2&0&q^4&-q\\
q^{-3}&q^{-4}&0&q^{-2}\\
-1&-q^{-1}&q^2&0
\end{pmatrix},
\label{eq:suppBplus}
\end{equation}
\begin{equation}
2B_0^-=
\begin{pmatrix}
0&q^{-2}&-q&-1\\
q^2&0&q^4&q^3\\
-q^{-1}&q^{-4}&0&-q^{-2}\\
-1&q^{-3}&-q^2&0
\end{pmatrix}.
\label{eq:suppBminus}
\end{equation}
For either branch, exact multiplication gives
\begin{equation}
 \operatorname{tr}B_0^\sigma=0,\qquad
 \operatorname{tr}(B_0^\sigma)^2=3,\qquad
 \operatorname{tr}(B_0^\sigma)^3=-\frac32(q+q^{-1}),
\label{eq:suppblocktraces}
\end{equation}
and
\begin{equation}
 \det B_0^\sigma=-\frac1{16}(q^2+1+q^{-2}).
\label{eq:suppblockdet}
\end{equation}
Newton's identities now give the characteristic polynomial, independently of merely knowing an annihilating polynomial:
\begin{align}
 p_q(\lambda)
 &=\lambda^4-\frac32\lambda^2+\frac{q+q^{-1}}2\lambda
 +\frac1{16}-\frac{(q+q^{-1})^2}{16}.
 \label{eq:suppcharpoly}
\end{align}
It agrees with the unrestricted annihilating polynomial
\eqref{eq:suppquartic}.
Writing $q=\ee^{\ii\theta}$,
\begin{equation}
 p_q(\lambda)=
 \left[(\lambda+\tfrac12)^2-\cos^2(\tfrac\theta2)\right]
 \left[(\lambda-\tfrac12)^2-\sin^2(\tfrac\theta2)\right].
 \label{eq:suppcharfactor}
\end{equation}
For every finite root in Eq.~\eqref{eq:suppgenericroot},
$\theta=\pi M/N\in(0,\pi)$ and the four roots are distinct:
\begin{equation}
 -\frac12\pm\cos\frac{\pi M}{2N},\qquad
 \frac12\pm\sin\frac{\pi M}{2N},
 \label{eq:supplocalroots}
\end{equation}
and each occurs once in every charge block, hence with total multiplicity $N$.
Because Eq.~\eqref{eq:suppSVMY} is a unitary similarity of the local seam,
the same four values and multiplicities hold for $S_\sigma^{\rm VMY}(M)$.  This
is a local spectrum statement and does not factorize either full-chain
Hamiltonian.

Finite clock--shift evaluations for coprime roots through $N=12$ provide an independent check of quotient aliases and transcription.  They are not used in place of the unrestricted proof of Eq.~\eqref{eq:suppquartic} or the analytic gauge transfer.

The basis in Eq.~\eqref{eq:supplocalbasis} correlates physical and endpoint labels.  It is not a tensor-product intertwiner compatible with the neighboring bulk bonds.  The local $N$-fold repetition therefore does not prove endpoint decoupling.

\section{Support-qualified uniqueness}
\label{app:uniqueness}

Let one coefficient tensor $K$ be copied translationally to $ABa$ and $BCa$ and obey
\begin{equation}
W_{Ba,\sigma}^\dagger K_{ABa}W_{Ba,\sigma}=K_{BCa}.
\label{eq:supphomogeneous}
\end{equation}
Choose $O_0=I/\sqrt2$ and three traceless Hilbert--Schmidt orthonormal operators on $A$.  The complete expansion
\begin{equation}
K_{ABa}=\sum_{\mu=0}^{3}O_\mu^A\otimes M_\mu^{Ba}
\end{equation}
has no structural restriction.  Since the right side of Eq.~\eqref{eq:supphomogeneous} is the identity on $A$, every traceless coefficient satisfies $W^\dagger M_\mu W=0$.  Invertibility gives
\begin{equation}
K_{ABa}=I_A\otimes R_{Ba}.
\label{eq:suppfirst}
\end{equation}
Translation of the same tensor gives $K_{BCa}=I_B\otimes R_{Ca}$, so
\begin{equation}
[W_{Ba}^\dagger R_{Ba}W_{Ba}]\otimes I_C=I_B\otimes R_{Ca}.
\end{equation}
The intersection
\begin{equation}
[\cB(Ba)\otimes I_C]\cap[I_B\otimes\cB(Ca)]
=I_B\otimes I_C\otimes\cB(a)
\end{equation}
forces $R_{Ba}=I_B\otimes r_a$.  This translated second-site step is logically distinct from Eq.~\eqref{eq:suppfirst}.  The residual equation is
\begin{equation}
[W_\sigma,I_2\otimes r_a]=0.
\end{equation}
The plus blocks require $[r_a,X^\dagger]=[r_a,Z]=0$; the minus blocks require $[r_a,X]=[r_a,Z]=0$.  Thus, in an arbitrary evaluated endpoint representation, the exact homogeneous kernel is
\begin{equation}
 K=I_A\otimes I_B\otimes r_a,\qquad r_a\in\{X,Z\}'.
 \label{eq:suppabstractkernel}
\end{equation}
Conversely every joint-commutant element solves the homogeneous equation.  Since the closed $S_\sigma$ is one particular solution, the abstract affine family in this support class is
\begin{equation}
 S_{\sigma,\mathrm{general}}
 =S_\sigma+I_A\otimes I_B\otimes r_a,\qquad r_a\in\{X,Z\}'.
 \label{eq:suppabstractaffine}
\end{equation}
For the standard irreducible $N$-dimensional clock--shift endpoint, distinct eigenvalues of $Z$ make $r_a$ diagonal, and the cyclic shift makes all diagonal entries equal.  The finite affine family therefore reduces to
\begin{equation}
 S_{\sigma,\mathrm{general}}=S_\sigma+cI_{4N},\qquad c\in\mathbb C.
 \label{eq:suppaffine}
\end{equation}
Hermiticity restricts $c$ to $\mathbb R$, and imposing the traceless
normalization selects $c=0$.

Conjugation by $F_{\sigma,B}$ is a bijection of the same support-class movement
solution spaces and Eq.~\eqref{eq:suppgaugecommutant} leaves the endpoint
commutant unchanged.  Therefore, for the source-normalized finite object,
\begin{equation}
 S_{\sigma,\mathrm{general}}^{\rm VMY}
 =S_\sigma^{\rm VMY}+cI_{4N},\qquad c\in\mathbb C,
 \label{eq:suppaffineVMY}
\end{equation}
with the same Hermitian and traceless-normalization qualifications.

This theorem is confined to the exact translated copy and support intersection above.  Wider supports, site-dependent tensors, changed endpoint representations, junction counterterms, and nonlocal deformations define different classification problems.  For reducible endpoints the affine freedom remains the full joint commutant in Eq.~\eqref{eq:suppabstractaffine}, rather than a scalar.

\paragraph{Finite-dimensional cross-check}
The full complex-linear homogeneous map was also assembled in the matrix-unit basis for $2\leq N\leq6$.  In each irreducible clock--shift representation its nullity is one, spanned by the identity, in agreement with Eq.~\eqref{eq:suppaffine}.  Independent dense evaluations check two-sided unitarity, Hermiticity, and local movement to machine precision.  These calculations validate transcription and small-$N$ aliases; the Laurent identities and the support theorem establish the all-$N$ statements.

\section{Branch exchange and projective antiunitary}
\label{app:global}

This section gives the complete finite-chain derivation behind the global theorem in the main text.  Here
\begin{equation}
 \begin{aligned}
 \eps&=\ee^{\ii\pi/N},&q&=\eps^M,&\omega&=q^2,\\
 \gcd(M,N)&=1,&N&\geq2,&L&\geq3.
 \end{aligned}
\end{equation}
and all site indices are modulo $L$.  Hats denote the source-normalized VMY
direct object; unhatted $W_\sigma,S_\sigma,H_{\rm Weyl}^\sigma$ denote the ungauged direct
Weyl object.  Neither notation denotes the endpoint-traced temporal MPO.

\subsection{Unified charge and ordered movement cycle}

Let
\begin{equation}
 Q=g^{\otimes L}\otimes Z^{-1},\qquad
 g=\operatorname{diag}(1,\omega).
\label{eq:suppglobalQ}
\end{equation}
The alternative orientation convention
$g_-=\operatorname{diag}(\omega,1)=\omega g^{-1}$ gives
\begin{equation}
 Q_-=g_-^{\otimes L}\otimes Z=\omega^LQ^{-1}.
\label{eq:suppQminus}
\end{equation}
The local selection rule \eqref{eq:suppselection} holds for both signs, so
\begin{equation}
 [H_{\rm Weyl}^\sigma(q;k),Q]=[H_{\rm VMY}^\sigma(M;k),Q]=0,
 \qquad \sigma=\pm.
\label{eq:suppbothcharge}
\end{equation}
Thus Eq.~\eqref{eq:suppQminus} is a second convention for one cyclic symmetry, not an independent charge.

Let $\widehat G_{\sigma,j}$ embed $W_\sigma^{\rm VMY}$ on physical site $j$ and the
retained endpoint.  We fix the product
\begin{equation}
 \widehat{\mathcal U}_\sigma(k)=
 \widehat G_{\sigma,k}\widehat G_{\sigma,k+1}\cdots\widehat G_{\sigma,k+L-1},
\label{eq:suppholonomy}
\end{equation}
where the rightmost factor acts first.  Repeated use of the local movement identity gives
\begin{equation}
 \widehat{\mathcal U}_\sigma(k)^\dagger H_{\rm VMY}^\sigma(M;k)
 \widehat{\mathcal U}_\sigma(k)
 =H_{\rm VMY}^\sigma(M;k+L)=H_{\rm VMY}^\sigma(M;k),
\end{equation}
and hence
\begin{equation}
 [\widehat{\mathcal U}_\sigma(k),H_{\rm VMY}^\sigma(M;k)]
 =[\widehat{\mathcal U}_\sigma(k),Q]=0.
\label{eq:suppholonomycomm}
\end{equation}
Adjacent starting points satisfy
\begin{equation}
 \widehat{\mathcal U}_\sigma(k+1)=
 \widehat G_{\sigma,k}^\dagger\widehat{\mathcal U}_\sigma(k)\widehat G_{\sigma,k}.
\end{equation}
The gauge-adjusted branch antiunitary $\widehat\Theta_L$ constructed below
obeys
\begin{equation}
 \widehat\Theta_L\widehat G_{\sigma,j}\widehat\Theta_L^{-1}
 =\omega^{-1}\widehat G_{-\sigma,j},
\end{equation}
and antiunitary conjugation does not reverse the multiplication order.  Therefore
\begin{equation}
 \widehat\Theta_L\widehat{\mathcal U}_\sigma(k)\widehat\Theta_L^{-1}
 =\omega^{-L}\widehat{\mathcal U}_{-\sigma}(k).
\label{eq:suppholonomytheta}
\end{equation}

Let $T$ translate states one site to the right,
$TO_jT^\dagger=O_{j+1}$ and $T^L=I$.  For a seam on $(L-1,0)$ define
\begin{equation}
 \widehat{\mathcal M}_\sigma=\widehat G_{\sigma,0}T.
\label{eq:suppM}
\end{equation}
Movement and translation covariance give
\begin{equation}
 [\widehat{\mathcal M}_\sigma,H_{\rm VMY}^\sigma(M;0)]
 =[\widehat{\mathcal M}_\sigma,Q]=0,\qquad
 \widehat{\mathcal M}_\sigma^L=\widehat{\mathcal U}_\sigma(0).
\label{eq:suppMpowers}
\end{equation}
No universal finite order of $\widehat{\mathcal U}_\sigma$ is asserted.  Removing
hats and output gauges gives the analogous same-sign holonomy for
$H_{\rm Weyl}^\sigma$, but not the all-root protected pairing result below.

\subsection{Endpoint Clifford shear and unitary branch exchange}

Let $R$ reverse the physical sites and fix the endpoint,
$RO_jR=O_{L-1-j}$.  Define
\begin{equation}
 D_q(t)=q^{t\sigma^z/2}
 =\operatorname{diag}\!\left(\ee^{\ii\pi Mt/(2N)},
 \ee^{-\ii\pi Mt/(2N)}\right),
\end{equation}
\begin{equation}
 \delta=N\bmod2,\qquad \xi=q^{1-\delta}.
\end{equation}
Let the endpoint inversion be
\begin{equation}
 J|n\rangle=|-n\!\!\pmod N\rangle,
 \label{eq:suppJ}
\end{equation}
with residue zero represented by $|N\rangle$.
The endpoint shear is diagonal:
\begin{equation}
 \begin{aligned}
 P_N|n\rangle&=p_n|n\rangle,\qquad p_1=1,\\
 p_{n+1}&=\xi\omega^np_n,\qquad
 n=1,\ldots,N,\qquad p_{N+1}=p_1.
 \end{aligned}
\label{eq:suppP}
\end{equation}
Its recursion closes because
\begin{equation}
 \xi^N\omega^{N(N+1)/2}=q^{N(N+2-\delta)}=1,
\end{equation}
where $N+2-\delta$ is even.  Thus the shear is valid for every coprime
$(M,N)$, including both values of $q^N=(-1)^M$.
Directly on the one-based endpoint basis,
\begin{equation}
 P_NXP_N^\dagger=\xi XZ,\qquad
 P_NZP_N^\dagger=Z,\qquad
 P_NJP_N^*J=Z^{-\delta}.
\label{eq:suppPrelations}
\end{equation}
These formulas include $N=2$ and do not treat aliased monomials as independent.

For any integer $b$, put
\begin{equation}
 a_0=-L-1,\qquad c_b=2L-\delta+2b,\qquad
 E_b=X^{a_0}Z^bP_N,
\label{eq:suppE}
\end{equation}
\begin{equation}
 V_b=\bigotimes_{j=0}^{L-1}D_q(c_b-2j),\qquad
 C_b=(V_b\otimes E_b)R.
\label{eq:suppC}
\end{equation}
Reflection reverses the orientation of every ordinary bond:
$Rh_{j,j+1}R=h^*$.  The adjacent twist exponents differ by two, and
\begin{equation}
 [D_q(t)\otimes D_q(t-2)]h^*
 [D_q(t)\otimes D_q(t-2)]^\dagger=h.
\label{eq:suppbulkrestore}
\end{equation}
For an endpoint monomial,
\begin{equation}
 E_bX^uZ^vE_b^\dagger
 =\xi^u\omega^{u(u-1)/2+bu-a_0(u+v)}
 X^uZ^{u+v}.
\label{eq:suppEmonomial}
\end{equation}
At the seam, the physical action is
\begin{equation}
 [D_q(2-\delta+2b)\otimes D_q(2L-\delta+2b)]P_{12}.
\end{equation}
Write $r=(r_1,r_2)$ and $c=(c_1,c_2)$ for the two binary spin strings labeling $E_{rc}$.  The net $q$ exponent multiplying a source row
$E_{rc}X^uZ^v$ after the physical swap, shear, and twists is
\begin{align}
 \kappa_b(r,c;u,v)
 ={}&(1-\delta)u+u(u-1)+2bu-2a_0(u+v)\nonumber\\
 &+(2-\delta+2b)(c_2-r_2)
 +(2L-\delta+2b)(c_1-r_1).
 \label{eq:suppkappa}
\end{align}
The $L$, $b$, and $\delta$ dependence cancels row by row against the selection-rule data.  Table~\ref{tab:branchexchange} records all twelve source rows.  We write the source and target coefficients as $2c_-=\varsigma_-q^{p_-}$ and $2c_+=\varsigma_+q^{p_+}$, with $\varsigma_\pm\in\{+1,-1\}$.  The column $\kappa$ is Eq.~\eqref{eq:suppkappa}, and the displayed exponent residual is
\begin{equation}
 \Delta p:=p_-+\kappa_b-p_+.
 \label{eq:suppbranchresidual}
\end{equation}
Thus every zero in the last column is a transparent rowwise equality of the source exponent plus the conjugation exponent with the target exponent; the signs also agree directly, $\varsigma_-=\varsigma_+$.

\begin{table*}[t]
\centering
\caption{Termwise phase comparison for $C_bS_-C_b^\dagger=S_+$.  Coefficients are decomposed as $2c_\pm=\varsigma_\pm q^{p_\pm}$, and the last column is explicitly $\Delta p=p_-+\kappa_b-p_+$.}
\label{tab:branchexchange}
\begin{tabular}{cccccccccc}
\hline
\multicolumn{4}{c}{branch $-$: source}&
\multicolumn{4}{c}{branch $+$: target}&
\multicolumn{2}{c}{residual}\\
\cline{1-4}\cline{5-8}\cline{9-10}
source & $(u,v)$ & $\varsigma_-$ & $p_-$ & target & $(u,u+v)$ & $\varsigma_+$ & $p_+$ & $\kappa_b$ & $\Delta p$\\
\hline
$E_{01}$ & $(-1,1)$ & $+1$ & $0$ & $E_{02}$ & $(-1,0)$ & $+1$ & $3$ & $3$ & $0$\\
$E_{02}$ & $(-1,0)$ & $-1$ & $1$ & $E_{01}$ & $(-1,-1)$ & $-1$ & $0$ & $-1$ & $0$\\
$E_{03}$ & $(-2,1)$ & $-1$ & $0$ & $E_{03}$ & $(-2,-1)$ & $-1$ & $4$ & $4$ & $0$\\
\cline{1-10}
$E_{10}$ & $(1,-1)$ & $+1$ & $-2$ & $E_{20}$ & $(1,0)$ & $+1$ & $-3$ & $-1$ & $0$\\
$E_{12}$ & $(0,-1)$ & $+1$ & $2$ & $E_{21}$ & $(0,-1)$ & $+1$ & $-2$ & $-4$ & $0$\\
$E_{13}$ & $(-1,0)$ & $+1$ & $3$ & $E_{23}$ & $(-1,-1)$ & $+1$ & $2$ & $-1$ & $0$\\
\cline{1-10}
$E_{20}$ & $(1,0)$ & $-1$ & $-1$ & $E_{10}$ & $(1,1)$ & $-1$ & $2$ & $3$ & $0$\\
$E_{21}$ & $(0,1)$ & $+1$ & $-2$ & $E_{12}$ & $(0,1)$ & $+1$ & $2$ & $4$ & $0$\\
$E_{23}$ & $(-1,1)$ & $-1$ & $-2$ & $E_{13}$ & $(-1,0)$ & $-1$ & $1$ & $3$ & $0$\\
\cline{1-10}
$E_{30}$ & $(2,-1)$ & $-1$ & $-4$ & $E_{30}$ & $(2,1)$ & $-1$ & $0$ & $4$ & $0$\\
$E_{31}$ & $(1,0)$ & $+1$ & $-3$ & $E_{32}$ & $(1,1)$ & $+1$ & $0$ & $3$ & $0$\\
$E_{32}$ & $(1,-1)$ & $-1$ & $0$ & $E_{31}$ & $(1,0)$ & $-1$ & $-1$ & $-1$ & $0$\\
\hline
\end{tabular}
\end{table*}

Equations~\eqref{eq:suppbulkrestore}--\eqref{eq:suppEmonomial} and the complete termwise comparison prove
\begin{equation}
 C_bS_-C_b^\dagger=S_+,\qquad
 C_bH_{\rm Weyl}^-(q;0)C_b^\dagger=H_{\rm Weyl}^+(q;0).
\label{eq:suppbranchunitary}
\end{equation}
For $b=0$, since the physical twists and reflection preserve
$g^{\otimes L}$ while $X^{a_0}Z^{-1}X^{-a_0}=\omega^{a_0}Z^{-1}$,
\begin{equation}
 C_0QC_0^\dagger=\omega^{-L-1}Q.
\label{eq:suppCcharge}
\end{equation}

The source-normalized branch unitary is not $C_0$ alone.  Define
\begin{equation}
 \widehat C_0=(F_+^{\otimes L}\otimes I_a)C_0.
 \label{eq:suppChat}
\end{equation}
Reflection moves the right-leg $F_-$ to the left seam site.  The uniform
$F_+^{\otimes L}$ preserves every ordinary bond, cancels that factor because
$F_+F_-=fI_2$, and leaves $F_+$ on the target right leg.  Hence
\begin{equation}
 \widehat C_0H_{\rm VMY}^-(M;0)\widehat C_0^\dagger
 =H_{\rm VMY}^+(M;0),\qquad
 \widehat C_0Q\widehat C_0^\dagger=\omega^{-L-1}Q.
 \label{eq:suppbranchunitaryVMY}
\end{equation}

\subsection{Gauge-adjusted same-branch antiunitary}

Let $\mathcal K_{z,a}$ be entrywise conjugation in the physical
$\sigma^z$ basis and the one-based endpoint basis, and use the endpoint
inversion in Eq.~\eqref{eq:suppJ}.  Choose the unique residue
and its compatible integer lift
\begin{equation}
 \kappa=(1-m_{\rm inv})\bmod N,\quad 0\leq\kappa<N,
 \qquad
 \tau=\frac{M(1-\kappa)-1}{N}\in\mathbb Z,
 \label{eq:suppkappatau}
\end{equation}
where $m_{\rm inv}=M^{-1}\bmod N$ is the residue defined in
Eq.~\eqref{eq:suppminv}.  Thus
\begin{equation}
 \begin{aligned}
 M(1-\kappa)&=1+N\tau,
 &\omega^\kappa&=(q/\eps)^2,
 &F_+^{-2}&=g^\kappa,\\
 F_+F_-&=fI_2,
 &\sigma^xF_+^*\sigma^x&=F_-^\dagger.&&
 \end{aligned}
 \label{eq:suppkappaphases}
\end{equation}
The compatible quotient in Eq.~\eqref{eq:suppkappatau} must not be replaced
by $(Mm_{\rm inv}-1)/N$ while retaining the displayed residue $\kappa$.

The complete antiunitary operators are
\begin{equation}
 \Theta_L^{\rm W}=[(\sigma^x)^{\otimes L}\otimes J]\mathcal K_{z,a},
 \qquad
 \widehat\Theta_L=
 [(\sigma^x)^{\otimes L}\otimes X^\kappa J]\mathcal K_{z,a}.
 \label{eq:suppTheta}
\end{equation}
Termwise conjugation gives the complete ungauged branch relations
\begin{equation}
 \begin{aligned}
 \Theta_L^{\rm W}H_{\rm Weyl}^+(q;0)(\Theta_L^{\rm W})^{-1}
 &=H_{\rm Weyl}^-(q;0),\\
 (\Theta_L^{\rm W})^2&=I,\\
 \Theta_L^{\rm W}Q(\Theta_L^{\rm W})^{-1}&=\omega^{-L}Q.
 \end{aligned}
 \label{eq:suppThetaWrelations}
\end{equation}
For the hatted operator, conjugation of
the output gauges initially produces $F_-^\dagger S_-F_-$.  Every row
$E_{rc}X^uZ^v$ of $S_-$ obeys
$v=-[n_B(r)-n_B(c)]$, and $X^\kappa$ contributes
\begin{equation}
 \omega^{-\kappa v}
 =(q/\eps)^{2[n_B(r)-n_B(c)]},
 \label{eq:suppThetaRowPhase}
\end{equation}
exactly the physical-row action of $F_-^2$.  Therefore
\begin{align}
 \widehat\Theta_LH_{\rm VMY}^+(M;0)\widehat\Theta_L^{-1}
 &=H_{\rm VMY}^-(M;0),\nonumber\\
 \widehat\Theta_LW_+^{\rm VMY}\widehat\Theta_L^{-1}
 &=\omega^{-1}W_-^{\rm VMY},\nonumber\\
 \widehat\Theta_L^2&=I,\qquad
 \widehat\Theta_LQ\widehat\Theta_L^{-1}=\omega^{\kappa-L}Q.
 \label{eq:suppThetaHatRelations}
\end{align}
The same twelve-row phase comparison verifies
Eq.~\eqref{eq:suppThetaRowPhase}; the gate relation follows by the same
four-block substitution.

Fix $b=0$ and define
\begin{equation}
 d=L+1-\delta,\qquad p=d+\kappa,\qquad c=2L+1-\kappa,
 \qquad
 \widehat{\mathcal A}=\widehat C_0\widehat\Theta_L.
 \label{eq:suppA}
\end{equation}
Then $\widehat{\mathcal A}$ is an antiunitary symmetry of
$H_{\rm VMY}^+(M;0)$ and
\begin{equation}
 \widehat{\mathcal A}Q\widehat{\mathcal A}^{-1}=\omega^{-c}Q.
 \label{eq:suppAcharge}
\end{equation}

We now multiply the complete antiunitary square.  Introduce
\begin{equation}
 \begin{aligned}
 R_\eps(t)&=\operatorname{diag}(1,\eps^t),\qquad k_j=k_0+2Mj,\\
 k_0&=1-M-M(2L-\delta).
 \end{aligned}
 \label{eq:suppkj}
\end{equation}
Up to an irrelevant unit scalar, the unitary part of
$\widehat{\mathcal A}=U_A\mathcal K_{z,a}$ is
\begin{equation}
 U_A=
 \left[\left(\bigotimes_{j=0}^{L-1}R_\eps(k_j)\right)
 (\sigma^x)^{\otimes L}R\right]
 \otimes[X^{-L-1}P_NX^\kappa J].
 \label{eq:suppUA}
\end{equation}
Reversal pairs $k_j$ with $k_{L-1-j}$.  Since
$\sigma^xR_\eps(-t)\sigma^x=\eps^{-t}R_\eps(t)$,
the physical factor is
\begin{equation}
 (U_AU_A^*)_{\rm phys}
 =\eps^{-\sum_jk_j}(g^{\otimes L})^{-p},\qquad
 -\sum_jk_j=ML(L+2-\delta)-L.
 \label{eq:suppA2physical}
\end{equation}
Here Eq.~\eqref{eq:suppkappatau} gives
$k_j+k_{L-1-j}=-2Mp-2N\tau$.  Put
$u=L+1+\kappa=p+\delta$.  The endpoint factor obtained from
Eq.~\eqref{eq:suppPrelations} is
\begin{equation}
 (U_AU_A^*)_a
 =\xi^u\omega^{u(u-1)/2-\kappa p}Z^p.
 \label{eq:suppA2endpoint}
\end{equation}
The complete scalar phase reduces exactly as
\begin{equation}
 \eps^{-\sum_jk_j}\xi^u
 \omega^{u(u-1)/2-\kappa p}
 =\eps^{Mpc+NL\tau}=(-1)^{L\tau}q^{pc}.
 \label{eq:suppA2phase}
\end{equation}
Consequently the full-space operator identity is
\begin{equation}
 \widehat{\mathcal A}^{2}=(-1)^{L\tau}q^{pc}Q^{-p}.
 \label{eq:suppAsquare}
\end{equation}

Although Eq.~\eqref{eq:suppkappatau} fixes a unique representative, this
formula is intrinsic.  Under $\kappa\mapsto\kappa+N$ one has
$\tau\mapsto\tau-M$, $p\mapsto p+N$, and $c\mapsto c-N$.  The ratio of
the transformed right side of Eq.~\eqref{eq:suppAsquare} to the original is
\begin{equation}
 (-1)^{-LM}q^{N(c-p-N)}
 =(-1)^{M(c-p-N-L)}=1,
 \label{eq:supprepresentative}
\end{equation}
because $c-p-N-L=\delta-N-2\kappa$ is even; also $Q^N=I$ and
$X^{\kappa+N}=X^\kappa$.  Thus the generators, sector map, and square are
representative independent.  The $\widehat{\mathcal B}^{2}$ formula below
is the right side of Eq.~\eqref{eq:suppAsquare} multiplied by
$(\eps\omega^L)^{r_0} Q^{-r_0}$, which contains neither $p$ nor $c$, so it
inherits the same representative invariance.

The co-moving calculation starts from the ungauged four-block identity,
which is algebraic at every coprime root:
\begin{align}
 &[D_q(2-\delta)\otimes E_0]\,W_-\,
 [D_q(-(2L-\delta))\otimes E_0^\dagger]\nonumber\\
 &\hspace{22mm}
 =q^{L+4}(g^{-1}\otimes Z)W_+^\dagger .
 \label{eq:suppfourblock}
\end{align}
The two $D_q$ factors act on the physical leg and $E_0$ on the endpoint.
Direct substitution of all four blocks of $W_-$ cancels the endpoint $X$
powers and gives, block by block,
\begin{equation}
 q^{L+4}(g^{-1}\otimes Z)W_+^\dagger
 =\frac{q^{L+4}}{\sqrt2}
 \begin{pmatrix}
 Z&-q^{-1}X^{-1}\\
 \omega^{-1}XZ&q^{-1}\omega^{-1}I
 \end{pmatrix}.
 \label{eq:suppfourblockexpanded}
\end{equation}
This explicit matrix records all four endpoint-valued blocks, rather than
only their common scalar.  Inserting the source gauges
multiplies the scalar by $f=\eps/q$: on the left $F_+F_-=fI_2$, while
on the right $W_+^\dagger F_+^\dagger=(W_+^{\rm VMY})^\dagger$.  Restoring
the reflected translation and every spectator twist therefore gives the
full-space identity

\begin{equation}
 \widehat{\mathcal A}\widehat{\mathcal M}_+
 \widehat{\mathcal A}^{-1}
 =(\eps\omega^L)Q^{-1}\widehat{\mathcal M}_+^{-1}.
\label{eq:suppAMA}
\end{equation}
The scalar is $fq^{2L+1}=\eps q^{2L}=\eps\omega^L$.  This is an operator
multiplication, not an inference from eigenvalues.

Choose
\begin{equation}
 r_0=\begin{cases}0,&L\ {\rm odd},\\1,&L\ {\rm even},\end{cases}
 \qquad
 \widehat{\mathcal B}=\widehat{\mathcal M}_+^{r_0}
 \widehat{\mathcal A}.
\label{eq:suppB}
\end{equation}
Using $[\widehat{\mathcal M}_+,Q]=0$ in
Eqs.~\eqref{eq:suppAsquare} and
\eqref{eq:suppAMA} gives
\begin{align}
 \widehat{\mathcal B}H_{\rm VMY}^+(M;0)
 \widehat{\mathcal B}^{-1}&=H_{\rm VMY}^+(M;0),\nonumber\\
 \widehat{\mathcal B}Q\widehat{\mathcal B}^{-1}
 &=\omega^{-c}Q,\nonumber\\
 \widehat{\mathcal B}^{2}
 &=(-1)^{L\tau}q^{pc}(\eps\omega^L)^{r_0} Q^{-(p+r_0)}.
\label{eq:suppBrelations}
\end{align}

\subsection{Anti-linearity, sector involution, and even multiplicity}

Let $Q|\psi_m\rangle=\omega^m|\psi_m\rangle$.  Because
$\widehat{\mathcal B}$ is anti-linear,
\begin{equation}
 \widehat{\mathcal B}Q|\psi_m\rangle
 =\omega^{-m}\widehat{\mathcal B}|\psi_m\rangle.
\end{equation}
On the other hand Eq.~\eqref{eq:suppBrelations} gives
$\widehat{\mathcal B}Q=\omega^{-c}Q\widehat{\mathcal B}$.  Therefore
\begin{equation}
 Q(\widehat{\mathcal B}|\psi_m\rangle)
 =\omega^{c-m}\widehat{\mathcal B}|\psi_m\rangle,
\end{equation}
so the exact involution is
\begin{equation}
 m\longmapsto c-m=2L+1-\kappa-m\pmod N.
\label{eq:suppsectormap}
\end{equation}

If $N$ is even, coprimality makes $M$ and $m_{\rm inv}$ odd.  Hence
$\kappa$ is even, $c$ is odd, and $2m=c\pmod N$ has no solution.  Every
sector is paired with a distinct orthogonal sector.

If $N$ is odd, there is exactly one fixed sector $m_*$.  Choose
$0\leq m_*<N$ and write
\begin{equation}
 c-2m_*=N\rho .
 \label{eq:supprho}
\end{equation}
Here $\delta=1$, $p=L+\kappa$, and parity gives
\begin{equation}
 \rho\equiv1-\kappa\pmod2,
 \qquad
 \tau\equiv M(1-\kappa)+1\pmod2.
 \label{eq:suppoddparities}
\end{equation}
Restricting the last line of Eq.~\eqref{eq:suppBrelations} to
$Q=\omega^{m_*}$ gives a sign $(-1)^E$ with
\begin{equation}
 E=L\tau+Mp\rho+r_0(M\rho-\tau).
 \label{eq:suppfixedexponent}
\end{equation}
Using Eq.~\eqref{eq:suppoddparities}, the first two terms reduce to $L$
modulo two and the coefficient of $r_0$ reduces to one.  Since
$r_0=(L+1)\bmod2$, $E$ is odd and
\begin{equation}
 \left.\widehat{\mathcal B}^{2}\right|_{m_*}=-I.
\label{eq:suppKramers}
\end{equation}
Kramers' argument supplies a linearly independent same-energy partner in
the fixed sector.  Distinct-sector pairing handles all other sectors.
Hence every eigenvalue of $H_{\rm VMY}^+(M;0)$ has even multiplicity.
Equation~\eqref{eq:suppbranchunitaryVMY} transfers the result to the minus
sign and Eq.~\eqref{eq:suppVMYmovement} to every seam position.  The theorem
states even multiplicity, not exact multiplicity two.

Because the Hamiltonian is finite dimensional and Hermitian, this even-multiplicity conclusion is equivalent to the existence of some nonunique spectral-basis unitary that writes it abstractly as $H_{\rm eff}\otimes I_2$.  Neither that linear-algebraic observation nor the proof above identifies the second factor with a fixed subspace of the retained endpoint, makes the factorization local, or makes it compatible with the movement gates and neighboring bonds.  The nontrivial statement here is the explicit charge-sector antiunitary with the square in Eq.~\eqref{eq:suppBrelations}.

\subsection{Proof of the relative-flux protection theorem}
\label{app:output-z2}

For the family in Eqs.~\eqref{eq:suppoutputgaugefamily}--
\eqref{eq:suppHeta}, choose the unique residue and its compatible lift
\begin{equation}
 \begin{aligned}
 M\kappa_\eta&\equiv-\eta\pmod N,\qquad 0\le\kappa_\eta<N,\\
 \tau_\eta&=-\frac{M\kappa_\eta+\eta}{N}\in\mathbb Z.
 \end{aligned}
 \label{eq:suppetakappatau}
\end{equation}
Then
\begin{equation}
 \begin{aligned}
 \omega^{\kappa_\eta}&=f_\eta^{-2},
 &(F_+^{(\eta)})^{-2}&=g^{\kappa_\eta},\\
 F_+^{(\eta)}F_-^{(\eta)}&=f_\eta I_2,
 &\sigma^x(F_+^{(\eta)})^*\sigma^x&=(F_-^{(\eta)})^\dagger.
 \end{aligned}
 \label{eq:suppetaphases}
\end{equation}
Define the branch operators at the reference seam by
\begin{equation}
 \mathcal C_\eta=
 \bigl((F_+^{(\eta)})^{\otimes L}\otimes I_a\bigr)C_0,
 \qquad
 \Theta_\eta=
 \bigl[(\sigma^x)^{\otimes L}\otimes X^{\kappa_\eta}J\bigr]
 \mathcal K_{z,a}.
 \label{eq:suppetabranchops}
\end{equation}
Reflection in $C_0$ moves the right-leg $F_-^{(\eta)}$ to the other seam
site.  The uniform $(F_+^{(\eta)})^{\otimes L}$ preserves all ordinary bonds,
cancels that factor, and leaves the target $F_+^{(\eta)}$.  Hence
\begin{equation}
 \mathcal C_\eta H_\eta^-(M;0)\mathcal C_\eta^\dagger=H_\eta^+(M;0).
 \label{eq:suppetabranchunitary}
\end{equation}
For every seam monomial $E_{rc}X^uZ^v$ of $S_-$, the twelve-term expansion
obeys $v=-\Delta n_B$, where
$\Delta n_B=n_B(r)-n_B(c)$.  Conjugation by $X^{\kappa_\eta}$ supplies
$\omega^{-\kappa_\eta v}=f_\eta^{-2\Delta n_B}$, exactly converting the
$F_-^{(\eta)\dagger}S_-F_-^{(\eta)}$ produced by the ungauged
antiunitary into
$F_-^{(\eta)}S_-F_-^{(\eta)\dagger}$.  Therefore
\begin{align}
 \Theta_\eta H_\eta^+(M;0)\Theta_\eta^{-1}
 &=H_\eta^-(M;0),\nonumber\\
 \Theta_\eta W_+^{(\eta)}\Theta_\eta^{-1}
 &=\omega^{-1}W_-^{(\eta)},\qquad \Theta_\eta^2=I.
 \label{eq:suppetaantiunitary}
\end{align}
The charge actions are
\begin{equation}
 \mathcal C_\eta Q\mathcal C_\eta^\dagger=\omega^{-L-1}Q,
 \qquad
 \Theta_\eta Q\Theta_\eta^{-1}=\omega^{\kappa_\eta-L}Q.
 \label{eq:suppetabranchcharge}
\end{equation}

Put
\begin{equation}
 \begin{aligned}
 \mathcal A_\eta&=\mathcal C_\eta\Theta_\eta,\qquad
 p_\eta=L+1-\delta+\kappa_\eta,\\
 c_\eta^+&=2L+1-\kappa_\eta,\qquad c_\eta:=c_\eta^+.
 \end{aligned}
 \label{eq:suppetaApc}
\end{equation}
It follows that $\mathcal A_\eta$ is a plus-branch antiunitary and
\begin{equation}
 \mathcal A_\eta Q\mathcal A_\eta^{-1}=\omega^{-c_\eta}Q,
 \qquad m\longmapsto c_\eta^+-m\pmod N.
 \label{eq:suppetasectormap}
\end{equation}

We next derive its square without a parity assumption.  Write
$\mathcal A_\eta=U_A\mathcal K_{z,a}$ and set
$d=L+1-\delta$ and $u=L+1+\kappa_\eta=p_\eta+\delta$.  Up to a unit scalar,
the physical diagonal factor in the unitary part of $\mathcal A_\eta$ is
\begin{equation}
 \mathcal V_\eta=\bigotimes_{t=0}^{L-1}
 \operatorname{diag}(1,\eps^{k_t}),\qquad
 k_t=\eta-M(2L-\delta)+2Mt.
 \label{eq:suppetaphysicaltwists}
\end{equation}
Reversal gives
\begin{equation}
 k_t+k_{L-1-t}=-2Mp_\eta-2N\tau_\eta,\qquad
 -\sum_{t=0}^{L-1}k_t=L(Md-\eta).
 \label{eq:suppetareversedtwists}
\end{equation}
Thus
\begin{equation}
 \begin{aligned}
 (U_AU_A^*)_{\rm phys}
 &=\eps^{L(Md-\eta)}(g^{\otimes L})^{-p_\eta},\\
 (U_AU_A^*)_a
 &=\xi^u\omega^{u(u-1)/2-\kappa_\eta p_\eta}Z^{p_\eta}.
 \end{aligned}
 \label{eq:suppetaA2factors}
\end{equation}
The scalar exponent reduces exactly:
\begin{align}
 &L(Md-\eta)
 +M\bigl[(1-\delta)u+u(u-1)-2\kappa_\eta p_\eta\bigr]
 \nonumber\\
 &\qquad=Mp_\eta c_\eta+NL\tau_\eta.
 \label{eq:suppetaA2reduction}
\end{align}
Since $\eps^N=-1$,
\begin{equation}
 \mathcal A_\eta^2
 =(-1)^{L\tau_\eta}q^{p_\eta c_\eta}Q^{-p_\eta}.
 \label{eq:suppetaAsquare}
\end{equation}

Let $\mathcal M_\eta=G_{+,\eta;0}T$, where $G_{+,\eta;0}$ embeds
$W_+^{(\eta)}$.  Exact movement gives
$[\mathcal M_\eta,H_\eta^+(M;0)]=[\mathcal M_\eta,Q]=0$.
Inserting the output factors into the four-block reflected-gate identity
multiplies the ungauged scalar by $f_\eta$ and changes no operator factor:
\begin{equation}
 \mathcal A_\eta\mathcal M_\eta\mathcal A_\eta^{-1}
 =\alpha_\eta Q^{-1}\mathcal M_\eta^{-1},\qquad
 \alpha_\eta=f_\eta q^{2L+1}.
 \label{eq:suppetareflectedmovement}
\end{equation}
For every integer $r$ define the plus-branch symmetry
\begin{equation}
 \mathcal B_{\eta,r}^+=\mathcal M_\eta^r\mathcal A_\eta,\qquad
 (\mathcal B_{\eta,r}^+)^2
 =(-1)^{L\tau_\eta}q^{p_\eta c_\eta}
  \alpha_\eta^rQ^{-(p_\eta+r)}.
 \label{eq:suppetaBsquare}
\end{equation}

If $m$ is fixed by Eq.~\eqref{eq:suppetasectormap}, choose
$\rho\in\mathbb Z$ so that $c_\eta-2m=N\rho$.  Restriction of
Eq.~\eqref{eq:suppetaBsquare} gives
\begin{equation}
 \left.(\mathcal B_{\eta,r}^+)^2\right|_m=(-1)^{E_{\eta,r}}I,\qquad
 E_{\eta,r}=L\tau_\eta+Mp_\eta\rho+r(M\rho-\tau_\eta).
 \label{eq:suppetafixedsquare}
\end{equation}
For even $N$, $M$ is odd and
$\kappa_\eta\equiv\eta\pmod2$.  If
$\eta\equiv M-1\pmod2$, then $\eta$ and $\kappa_\eta$ are even, so
$c_\eta$ is odd and no fixed sector exists.  For odd $N$ the fixed sector
is unique, and exact parity reduction gives
\begin{equation}
 E_{\eta,r}\equiv(M+\eta)(L+r)\pmod2.
 \label{eq:suppetaoddparity}
\end{equation}
The protected condition makes $M+\eta$ odd; choosing
$r=r_0=(L+1)\bmod2$ therefore gives square $-I$.  This proves the
protected implication for the plus branch.

For the original $Q$ labels of the minus branch, define
\begin{equation}
 \mathcal B_{\eta,r}^-=
 \mathcal C_\eta^\dagger\mathcal B_{\eta,r}^+\mathcal C_\eta,
 \qquad
 c_\eta^-=c_\eta^+-2(L+1)=-1-\kappa_\eta\pmod N .
 \label{eq:suppetaminusmap}
\end{equation}
Equation~\eqref{eq:suppetabranchcharge} shows that $\mathcal C_\eta$ shifts a
charge label by $L+1$; hence $\mathcal B_{\eta,r}^-$ maps
$m\mapsto c_\eta^--m$.  Unitary conjugation transfers the fixed-sector square
and its multiplicity conclusion to this relabelled minus-branch sector.

The same formulas sharply delimit this mechanism in the other coset.  At
odd $N$, Eq.~\eqref{eq:suppetaoddparity} is even for every $r$.  At even
$N$ there are two fixed sectors; their exponents differ by
$M(p_\eta+r)$ modulo two, so the two squares are either both $+I$ or have
opposite signs.  At least one fixed sector is therefore not Kramers
protected by any $\mathcal B_{\eta,r}^+$; branch conjugation gives the same
conclusion for $\mathcal B_{\eta,r}^-$.  Hence the analytic antiunitary
construction does not protect every fixed sector in the other coset.
\ref{app:fluxscan} tests the corresponding spectral pattern on a
fixed finite-size grid.

Finally, the formulas do not depend on auxiliary integer representatives.
The replacement
\begin{equation}
 \begin{aligned}
 \kappa_\eta&\mapsto\kappa_\eta+N,
 &\tau_\eta&\mapsto\tau_\eta-M,\\
 p_\eta&\mapsto p_\eta+N,
 &c_\eta&\mapsto c_\eta-N
 \end{aligned}
 \label{eq:suppetarepresentative}
\end{equation}
does not change a generator because $X^N=I$.  The ratio of the transformed
right-hand side of Eq.~\eqref{eq:suppetaAsquare} to the original is
\begin{equation}
 (-1)^{-LM}q^{N(c_\eta-p_\eta-N)}
 =(-1)^{M(c_\eta-p_\eta-N-L)}=1,
\end{equation}
because $c_\eta-p_\eta-N-L=\delta-N-2\kappa_\eta$ is even; the same
ratio controls Eq.~\eqref{eq:suppetaBsquare}.  Also
$\eta\mapsto\eta+2N$ leaves the physical factors and $\kappa_\eta$
unchanged and sends $\tau_\eta\mapsto\tau_\eta-2$.  Hence the result is
intrinsic to $\eta\in\mathbb Z_{2N}$.

\section{Exact finite controls and boundary diagnostics}
\label{app:controls}

\subsection{Ungauged protection boundary and exact counterexample}

The ungauged control uses fully specified operators.  Let
$G_{+,j}^{\rm W}=W_{+,ja}$ and define
\begin{equation}
 \mathcal M_+^{\rm W}=G_{+,0}^{\rm W}T,\qquad
 \mathcal A_{\rm W}=C_0\Theta_L^{\rm W},\qquad
 \mathcal B_{\rm W}=(\mathcal M_+^{\rm W})^{r_0}\mathcal A_{\rm W},
 \label{eq:suppBWdefinition}
\end{equation}
with the same $r_0$ as Eq.~\eqref{eq:suppB}.  The ungauged movement,
branch-unitary, and branch-antiunitary relations make
$\mathcal B_{\rm W}$ an antiunitary symmetry of
$H_{\rm Weyl}^+(q;0)$.  Directly specializing the preceding multiplication
by removing the output gauges gives
\begin{align}
 \mathcal A_{\rm W}Q\mathcal A_{\rm W}^{-1}
 &=\omega^{-(2L+1)}Q,\nonumber\\
 \mathcal A_{\rm W}^{2}
 &=q^{d(2L+1)}Q^{-d},\nonumber\\
 \mathcal A_{\rm W}\mathcal M_+^{\rm W}\mathcal A_{\rm W}^{-1}
 &=q^{2L+1}Q^{-1}(\mathcal M_+^{\rm W})^{-1},\nonumber\\
 \mathcal B_{\rm W}^{2}
 &=q^{(d+r_0)(2L+1)}Q^{-(d+r_0)}.
 \label{eq:suppBWfullspace}
\end{align}
Anti-linearity and the first line give the sector map
\begin{equation}
 m\longmapsto2L+1-m\pmod N.
 \label{eq:suppWeylsectormap}
\end{equation}
Even $N$ has no fixed sector because $2L+1$ is odd.  For odd $N$, let
$2L+1-2m_*=N\rho_{\rm W}$.  Then both $\rho_{\rm W}$ and
$d+r_0$ are odd, so the last line of Eq.~\eqref{eq:suppBWfullspace}
restricts to
\begin{equation}
 \left.\mathcal B_{\rm W}^{2}\right|_{m_*}
 =q^{N(d+r_0)\rho_{\rm W}}I=(-1)^M I.
 \label{eq:suppWeylfixedsquare}
\end{equation}
Thus this explicit mechanism protects the $q^N=-1$ branch (odd $M$), while
on the $q^N=+1$ branch (even $M$) it gives no Kramers partner.  The latter
statement concerns failure of the mechanism, not an assertion that every
even-$M$ finite size has a simple level.

The all-parameter ungauged extension is nevertheless false exactly.  Set
\begin{equation}
 (N,M,L)=(3,2,3),\qquad q^2+q+1=0.
\end{equation}
For either branch, exact arithmetic over $\mathbb Q[q]/(q^2+q+1)$ gives
\begin{equation}
 \det(\lambda I-H_{\rm Weyl}^{\pm})=A(\lambda)B(\lambda)^2,
 \label{eq:suppcounterfactor}
\end{equation}
where
\begin{align}
256A(\lambda)={}&256\lambda^8-1920\lambda^6-512\lambda^5
+3504\lambda^4+1152\lambda^3\nonumber\\
&-2072\lambda^2-480\lambda+393,
\label{eq:suppcounterA}\\[1mm]
256B(\lambda)={}&256\lambda^8-1920\lambda^6-896\lambda^5
+3504\lambda^4+3168\lambda^3\nonumber\\
&+328\lambda^2-264\lambda-39.
\label{eq:suppcounterB}
\end{align}
The square-free and Euclidean calculations give
\begin{equation}
 \gcd(A,A')=\gcd(B,B')=\gcd(A,B)=1.
 \label{eq:suppcountergcd}
\end{equation}
Thus the degree-eight $A$ factor contributes eight simple real energy
levels, whereas the degree-eight $B$ factor contributes eight double
levels.  Hermiticity makes all roots real, and exact recomposition gives the
full degree-24 characteristic polynomial.  This calculation concerns $H_{\rm Weyl}$ and does not
contradict the all-size theorem for $H_{\rm VMY}$.

\subsection{Exact source-normalized finite-point factorization}
\label{app:exactvmyfactor}

At the same $(N,M,L)=(3,2,3)$ point,
$\eps=1+q$ and $f=\eps/q=-q$ in
$\mathbb Q(q)/(q^2+q+1)$.  The literal source gauges are therefore
\begin{equation}
 F_+=\operatorname{diag}(1,-q),\qquad
 F_-=\operatorname{diag}(-q,1).
 \label{eq:suppexactsourcegauges}
\end{equation}
We apply them to the second physical leg of the exact seam, assemble both
$24\times24$ source-normalized Hamiltonians, and compute their
characteristic polynomials using rational-pair arithmetic.  No floating
diagonalization or eigenvalue clustering enters.  Both branches give
$\chi_{\rm VMY}^{\sigma}=C^2$, where
\begin{align}
4096C(\lambda)={}&4096\lambda^{12}-46080\lambda^{10}
-18432\lambda^9\nonumber\\
&+170496\lambda^8+129024\lambda^7-206400\lambda^6\nonumber\\
&-229248\lambda^5+17568\lambda^4+75328\lambda^3\nonumber\\
&+14220\lambda^2-3096\lambda-459,
\label{eq:suppCpolynomial}
\end{align}
and the exact Euclidean algorithm gives $\gcd(C,C')=1$.  The uniform
comparison has
\begin{equation}
 \chi_0(\lambda)
 =(\lambda-\tfrac94)^6(\lambda+\tfrac34)^{18},
 \qquad H_0=H_{\rm unif}\otimes I_3.
 \label{eq:suppuniformfactor}
\end{equation}
Together with Eqs.~\eqref{eq:suppcounterfactor}--
\eqref{eq:suppcountergcd}, the source factorization proves the main-text
separation of the two relative-holonomy orbits: eight simple plus eight double
levels versus twelve double levels.  Equation~\eqref{eq:suppuniformfactor} is
retained only as an auxiliary reference to the unmodified chain with a
spectator endpoint.  The accompanying machine-readable data record the full
polynomial coefficients, square-free recompositions, and basis conventions.

The finite-point factors and the all-size second-moment lemma have
complementary roles.  The former distinguishes the two direct
normalizations at one exact point.  The latter separates either direct
normalization from $H_{\rm unif}\otimes I_N$ for all allowed roots and
lengths, but cannot distinguish $H_{\rm Weyl}$ from $H_{\rm VMY}$ because
their second moments coincide.

\subsection{Sector-resolved finite-size test of the spectral converse}
\label{app:fluxscan}

We independently reconstructed and charge-block diagonalized every coprime
root with $2\le N\le12$, all $L=3,4,5$, both seam branches, and one
representative of each output-parity orbit.  Exact
Eq.~\eqref{eq:suppetaplustwo} makes additional values in the same parity orbit
unitarily redundant.  Table~\ref{tab:scan-summary} summarizes the
resulting $540$ Hamiltonians; the numerical diagnostics and exact controls
underlying that summary are given below.

The maximum Hermiticity residual is $1.15\times10^{-14}$, the maximum
paired-sector spectral residual is $1.60\times10^{-14}$, and the largest
diameter of any degeneracy cluster is $1.87\times10^{-14}$.  The tightest
case, $(N,M,L,\sigma,\eta)=(10,1,5,-,0)$ and its isospectral plus branch, has
an 80-digit charge-block estimate
$6.06334\times10^{-7}$ for the nearest distinct-doublet gap; this value was
checked using 80-digit working arithmetic, and the Jacobi off-diagonal
Frobenius residual is below $6.6\times10^{-62}$.  The full-precision value and
diagnostics are retained in the repository package described in the Data
availability statement.
Thus the classification is separated from numerical cluster widths by more
than seven orders of magnitude, while the main clustering tolerance itself
is reported separately.  Every protected case nevertheless contains simple
levels inside at least one individual nonfixed sector.  Full-space evenness
must therefore not be described as uniform within-sector Kramers degeneracy.

The two exceptional even-$N$ rows in the table are the branches of
$(N,M,L,\eta)=(2,1,4,0)$.  Exact arithmetic over $\mathbb Q(i)$ gives the
same characteristic polynomial in the two charge sectors,
\begin{equation}
 \begin{aligned}
 P_{\rm sec}(\lambda)
 &=\left(\lambda^2-\frac34\right)^2R_{12}(\lambda),\\
 \gcd(R_{12},R_{12}')
 &=\gcd\!\left(R_{12},\lambda^2-\frac34\right)=1.
 \end{aligned}
 \label{eq:suppn2protected}
\end{equation}
Hence $\lambda=\pm\sqrt3/2$ is double in each sector and fourfold in the
full Hilbert space; the remaining twelve energy levels are double.  The
full-space multiplicity histogram is therefore $2{:}12;4{:}2$.  At the
unprotected representative $\eta=1$, the exact sector factorization together
with cross-sector coprimality gives eight simple full-space levels.  These exact controls show
both why higher protected even multiplicities are allowed and why the
opposite-orbit scan is plausible, without converting the scan into a proof.
The largest full-space multiplicity is four in the protected class (268 cases
at maximum two, two cases at maximum four) and six in the opposite class (268
cases at maximum two, two cases at maximum six).  The latter exceptions are
the two branches of $(N,M,L,\eta)=(2,1,5,1)$, with histogram
$1{:}8;2{:}16;6{:}4$.  Hence ``unprotected'' does not mean
``nondegenerate.''

\subsection{Exact failure of naive transfer-family inheritance}

Use the VMY cyclic tensor reproduced in Eq.~\eqref{eq:suppvmyLax}, and
take the exact finite point
\begin{equation}
 (N,M,L)=(3,1,3),\qquad
 (\ee^u,\ee^v,\ee^s)=(2,1,1).
 \label{eq:supptransferpoint}
\end{equation}
Let $b$ denote the traced cyclic auxiliary space and fix the source order
\begin{equation}
 T_c=\Tr_b\!\left[
 \mathcal L_{b,1}\mathcal L_{b,2}\mathcal L_{b,3}\right].
 \label{eq:supporderedmonodromy}
\end{equation}
The displayed product is the left-to-right monodromy order used here, so
$\mathcal L_{b,3}$ acts first on a ket.  We order the physical basis as
$|s_1s_2s_3\rangle$ and the retained-endpoint basis as
$|s_1s_2s_3\rangle\otimes|n\rangle_a$, with the endpoint index $n$ fastest;
the zero-based flat index is
$3(4s_1+2s_2+s_3)+(n-1)$.  We use $[A,B]=AB-BA$.
All arithmetic lies in $\mathbb Q[q]/(q^2-q+1)$.  Entrywise exact
calculation first supplies the source control
\begin{equation}
 [H_{\rm unif}\otimes I_3,T_c\otimes I_3]=0,
\end{equation}
but for the retained-endpoint direct seam gives
\begin{equation}
 [H_{\rm VMY}^+(1;0),T_c\otimes I_3]\neq0,
 \qquad
 [H_{\rm VMY}^-(1;0),T_c\otimes I_3]\neq0.
 \label{eq:supptransfercounter}
\end{equation}
Each nonzero commutator has 162 nonzero entries.  Representative exact
entries (with zero-based matrix indices) and complete
squared Frobenius norms are
\begin{align}
 ([H_{\rm VMY}^+,T_c\otimes I_3])_{0,4}
 &=\frac{159}{16}+3q,\nonumber\\
 \|[H_{\rm VMY}^+,T_c\otimes I_3]\|_F^2
 &=\frac{620487}{64},\nonumber\\
 ([H_{\rm VMY}^-,T_c\otimes I_3])_{0,4}
 &=\frac92+\frac{147}{16}q,\nonumber\\
 \|[H_{\rm VMY}^-,T_c\otimes I_3]\|_F^2
 &=\frac{639927}{64}.
 \label{eq:supptransferexact}
\end{align}
As an order-sensitivity check, replacing this monodromy by
$\Tr_b(\mathcal L_{b,3}\mathcal L_{b,2}\mathcal L_{b,1})$ interchanges the
two squared norms.  The product order is therefore part of the exact
definition, not a cosmetic convention.
At $M=1$, $F_+=F_-=I_2$, so there is no Weyl/VMY normalization ambiguity
at this test point.  The exact calculation establishes failure of inheritance
of this \emph{undressed uniform} cyclic member by the direct seam.  Dressed,
endpoint-bearing, mixed-intertwiner, junction-dependent, and other
seam-specific transfer families---as well as the broader integrability
question---remain open.

\subsection{Computational cross-checks}

The preceding proof is analytic and does not use diagonalization.  Independent finite-dimensional implementations nevertheless reconstruct the branch exchange, branch antiunitary, charge action, co-moving identity, full-space squares, and representative invariance for coprime roots with $2\leq N\leq7$ and $L=3,4,5$.  Direct diagonalization on the same finite set agrees with the protected even-multiplicity theorem and includes the unprotected exact counterexample above.  A separate integer-arithmetic implementation checks the modular lifts and fixed-sector parity over a substantially larger range of $N$ and $L$.  These calculations are transcription and finite-size checks only; the all-size conclusion follows from Eqs.~\eqref{eq:suppetaAsquare}--\eqref{eq:suppetaoddparity}.

\subsection{Relation to prior work and operator-domain boundaries}
\label{app:prior-boundaries}

The Discussion gives the physical literature positioning.  The operator-domain
boundary needed for the proofs is narrower: VMY's cyclic transfer and temporal
duality operators are auxiliary-traced, whereas
Eqs.~\eqref{eq:suppHWeyl}--\eqref{eq:suppHVMY} retain the endpoint as a
dynamical factor \citep{Vernier2026Lattice}.  Ueda et al.\ provide the nearest
dangling-virtual impurity framework, but their Ising doubling mechanism is a
unitary spectator factorization \citep{Ueda2026Perfect}.  Related mobility,
antiunitary, and movement-cycle results concern different operators or
sectors \citep{Liang2026Physical,Aasen2020Dualities,Sato2026Invariants}.

Likewise, loop-algebra multiplets, twisted transfer matrices,
Fabricius--McCoy strings, $Q$-operator towers, and hidden twisted sectors are
sector-, tower-, or selected-eigenspace statements
\citep{DeguchiFabriciusMcCoy2001Loop,Deguchi2004Twisted,Korff2004Twisted,Miao2021QOperator,Hu2026Hidden}.
These results address operator domains distinct from the charge-sector
antiunitary proved here.

\subsection{Secondary second-moment boundary diagnostic}
\label{app:moment}

This trace comparison is retained as a narrow boundary diagnostic and is not used in the Weyl, local-spectrum, holonomy, or even-multiplicity theorems.  It applies to the entire output family and every coprime $(M,N)$, but rules out only one naive equal-size comparison with the uniform chain tensored by the displayed endpoint identity.

Let
\begin{equation}
 \theta=\frac{\pi M}{N},\qquad
 \Delta=\frac{q+q^{-1}}2=\cos\theta.
 \label{eq:suppmomenttheta}
\end{equation}
In the ordered two-spin basis, the XXZ density is
\begin{equation}
h=\begin{pmatrix}
\tfrac12\cos\theta&0&0&0\\
0&-\tfrac12\cos\theta&q&0\\
0&q^{-1}&-\tfrac12\cos\theta&0\\
0&0&0&\tfrac12\cos\theta
\end{pmatrix},
\label{eq:hmatrix}
\end{equation}
where the diagonal entries follow from $(q+q^{-1})/4=\cos\theta/2$ and the eigenvalues of $\sigma^z\otimes\sigma^z$.  Direct multiplication gives
\begin{equation}
\Tr h^2=2+\cos^2\theta.
\label{eq:supptrh2}
\end{equation}
All one-site partial traces of $h$ vanish.

The ungauged seam has six Hermitian off-diagonal pairs.  For one pair $cE_{rc}Y+c^*E_{cr}Y^\dagger$, squaring and tracing gives $2|c|^2N$.  Here $|c|=1/2$ for all six pairs, so
\begin{equation}
\Tr S_\sigma^2=6\cdot2\cdot\frac14N=3N.
\end{equation}
Every physical-output similarity in Eq.~\eqref{eq:suppoutputgaugefamily}
preserves this trace, so $\Tr[(S_\sigma^{(\eta)})^2]=3N$ for all $\eta$.

For later use, the finite Weyl trace is
\begin{equation}
\Tr(Y_{m,n})=N\,\delta_{m,0\ ({\rm mod}\ N)}\delta_{n,0\ ({\rm mod}\ N)}.
\end{equation}
Every endpoint monomial in the seam is traceless for $N\geq2$.  At $N=2$, the only exponent $m=\pm2$ that becomes zero modulo $N$ is paired with $n=\pm1$, so it remains traceless.  Consequently
\begin{equation}
\Tr_aS_\sigma=\Tr_aS_\sigma^{(\eta)}=0.
\label{eq:partialseamtrace}
\end{equation}

For periodic $L\geq3$, two distinct ordinary cycle edges are disjoint or share exactly one site.  Disjoint cross traces factor through $\Tr h=0$.  If $h_{uv}$ and $h_{vw}$ share $v$, tracing the exterior sites first gives
\begin{equation}
\Tr_{uvw}(h_{uv}h_{vw})
=\Tr_v\!\left[(\Tr_u h_{uv})(\Tr_w h_{vw})\right]=0.
\end{equation}
This covers the triangle $L=3$, where all distinct edge pairs meet.  For any physical operator $B$, Eq.~\eqref{eq:partialseamtrace} gives
\begin{equation}
\Tr_{\rm phys,a}[S_\sigma(B\otimes I_a)]
=\Tr_{\rm phys}[(\Tr_aS_\sigma)B]=0,
\end{equation}
so every seam--bond cross term vanishes whether or not the physical supports overlap.  Write $H_\eta^\sigma$ for any retained-endpoint member in Eq.~\eqref{eq:suppHeta} at any seam position (the trace is position independent).  Define, before using it,
\begin{equation}
 \Delta_{2,\eta}^{\sigma}(N,M,L):=
 \Tr[(H_{\eta}^{\sigma})^2]
 -\Tr[(H_{\rm unif}\otimes I_N)^2].
 \label{eq:suppDelta2definition}
\end{equation}
Including the $2^{L-2}$ trace over spectator spins yields
\begin{align}
\Tr[(H_{\rm unif}\otimes I_N)^2]
&=NL2^{L-2}(2+\cos^2\theta),\nonumber\\
\Tr[(H_{\eta}^\sigma)^2]
&=N(L-1)2^{L-2}(2+\cos^2\theta)+3N2^{L-2},\nonumber\\
\Delta_{2,\eta}^{\sigma}(N,M,L)
 &=2^{L-2}N\sin^2\!\left(\frac{\pi M}{N}\right)>0,
 \qquad L\geq3.
\label{eq:suppmoment}
\end{align}

Equation~\eqref{eq:suppmoment} separates the displayed direct
$H_{\eta}^\sigma$ from $H_{\rm unif}\otimes I_N$ in the equal-$N$,
equal-$L$, standard-trace comparison.  Factorizations with a different
$H_{\rm eff}$, or dimension-changing, projected, quotient, junction,
sectorwise, dressed-transfer, and infrared maps, lie outside this trace
comparison.  Its analytic range is $L\geq3$.

The same formula must not be extended to $L=2$.  There the two directed periodic bonds occupy the same two-site support.  For $\mathfrak f\in\{\mathrm{Weyl},\mathrm{VMY}\}$ and with $\Delta=\cos(\pi M/N)$, direct calculation instead gives
\begin{align}
\Tr[(H_{\rm unif}^{(L=2)}\otimes I_N)^2]&=12N\Delta^2,\\
\Tr[(H_{\mathfrak f}^{\sigma,(L=2)})^2]&=N(5+\Delta^2),
\end{align}
so their difference is $N(5-11\Delta^2)$ rather than Eq.~\eqref{eq:suppmoment}.  The local movement theorem also requires three distinct sites and is not asserted at $L=2$.

\subsection{Machine-verification map}
\label{app:repro}

The accompanying code contains exact-arithmetic and finite-dimensional
implementations of the derivations above.  The unrestricted calculations
represent coefficients as Gaussian-rational Laurent polynomials in $q$ and
normal-order every endpoint monomial before any finite-root quotient.  Finite
clock--shift calculations use the one-based endpoint basis
$|1\rangle,\ldots,|N\rangle$, physical-major/endpoint-minor tensor ordering,
and row-major vectorization.  Exponents are reduced modulo $N$ only after the
unrestricted movement and quartic residuals have been collected.

The code reproduces the 240-variable coefficient system, the complete
movement residual, the quartic identity, the homogeneous commutant, the
branch phase comparison, the full-space antiunitary identities, and the
fixed-sector parity reduction.  It also reconstructs the exact
characteristic polynomials at $(N,M,L)=(3,2,3)$ over
$\mathbb Q[q]/(q^2+q+1)$ and the transfer commutators at $(3,1,3)$ over
$\mathbb Q[q]/(q^2-q+1)$, using the monodromy and basis conventions in
\ref{app:global}.  Floating-point diagonalization is used only for
finite corroborative checks, with residuals normalized by the Frobenius norm;
it is not used to establish any all-size theorem.

\bibliographystyle{elsarticle-num} 
\bibliography{ref}

\end{document}